\documentclass[10pt,twocolumn,twoside,english]{IEEEtran}
\usepackage{hyperref}
\usepackage{amssymb}
\usepackage{amsmath}
\usepackage{graphicx,wrapfig,times,amsthm,bm,amsmath,amsfonts,amssymb}
\usepackage{subfig}
\usepackage{threeparttable}
\usepackage[wide]{sidecap}
\usepackage{array}
\usepackage{algorithm}
\usepackage{algpseudocode}
\usepackage{lipsum,babel}
\usepackage{color}
\usepackage{setspace}

\usepackage{tikz}
\usetikzlibrary{arrows.meta}
\tikzset{%
	>={Latex[width=2mm,length=2mm]},
	base/.style = {rectangle, rounded corners, draw=black,
		minimum width=0.3cm, minimum height=0.3cm,
		text centered, font=\texttrademark},
	OptimalPolicies/.style= {base, text=green},
	Estimation/.style= {base, text=red},
	Stabilization/.style= {base, text=blue},
	RLAlgo/.style= {base, text=violet},
}

\newtheorem{thm}{Theorem}
\newtheorem{lem}{Lemma}
\newtheorem{deff}{Definition}

\newtheorem{assum}{Assumption}
\newtheorem{cor}{Corollary}

\newtheorem{remark}{Remark}

\def \det {\mathrm{det}}

\def \exp {\mathrm{exp}}
\def \cov {\mathrm{Cov}}
\def \var {\boldsymbol{\mathrm{Var}}}

\def \R {\mathbb{R}}
\def \C {\mathbb{C}}
\def \tailconst {b_1}
\def \tailcoeff {b_2}
\def \tailexp {\alpha}

\def \cost {\mathcal{J}}
\def \para {\theta}
\def \rrate {\gamma}
\def \ssconstant {\rho}
\def \firstcoeff {6}

\def \policy {\pi}

\newcommand{\Mnorm}[2]{{\left\vert\kern-0.35ex\left\vert\kern-0.35ex\left\vert #1 
		\right\vert\kern-0.35ex\right\vert\kern-0.35ex\right\vert}_{#2}}
\newcommand{\Opnorm}[3]{{\left\vert\kern-0.35ex\left\vert\kern-0.35ex\left\vert #1 
		\right\vert\kern-0.35ex\right\vert\kern-0.35ex\right\vert}_{#2 \to #3}}
\newcommand{\norm}[2]{{\left\vert\kern-0.35ex\left\vert #1 
		\right\vert\kern-0.35ex\right\vert}_{#2}}
\newcommand{\eigmax}[1]{\left| \lambda_{\max} \left( #1 \right)\right|}
\newcommand{\eigmin}[1]{\left| \lambda_{\min} \left( #1 \right)\right|}
\newcommand{\tr}[1]{\boldsymbol{\mathrm{tr}} \left( #1 \right)}
\newcommand{\PP}[1]{\mathbb{P} \left(#1\right)}
\newcommand{\E}[1]{\boldsymbol{\mathbb{E}} \left[#1\right]}

\newcommand{\MJordanconst}[2]{\boldsymbol{\eta}_{#1}\left(#2\right)}
\newcommand{\mult}[1]{\mu_{#1}}
\newcommand{\samplesize}[3]{{N}_{}\left(#2,#3\right)}
\newcommand{\event}[1]{{\mathcal{#1}}}
\newcommand{\noisemax}[2]{\boldsymbol{\nu}_{#1}\left(#2\right)}

\newcommand{\diag}[1]{\boldsymbol{\mathrm{diag}}\left(#1\right)}

\newcommand{\innerproductminconstant}[1]{\psi_0}

\newcommand{\regret}[1]{\mathcal{R} \left(#1\right)}
\newcommand{\loss}[3]{\mathcal{L}_#1^{#2} \left(#3\right)}
\newcommand{\symmetrizer}[1]{\Phi \left(#1\right)}
\newcommand{\dimension}[2]{\mathrm{dim}_{#1} \left(#2\right)}
\newcommand{\Lipschitz}[1]{\chi_{_{#1}}}
\newcommand{\paraspace}[1]{\Omega^{(#1)}}
\newcommand{\tempparaspace}[1]{\Gamma^{(#1)}}
\newcommand{\Kmatrix}[1]{K\left(#1\right)}
\newcommand{\Lmatrix}[1]{L\left(#1\right)}
\newcommand{\extendedLmatrix}[1]{\widetilde{L}\left(#1\right)}
\newcommand{\avecost}[2]{\overline{\mathcal{J}}_{#2}\left(#1\right)}
\newcommand{\optcost}[1]{{\mathcal{J}^\star}\left(#1\right)}
\newcommand{\instantcost}[1]{c_{#1}}

\newcommand{\predbound}[2]{\boldsymbol{\xi}_{#1}\left(#2\right)}
\newcommand{\prediction}[2]{\boldsymbol{\beta}_{#1}\left(#2\right)}
\newcommand{\order}[1]{\mathcal{O} \left(#1\right)}
\newcommand{\orderlog}[1]{\widetilde{\mathcal{O}} \left(#1\right)}
\newcommand{\optpara}[1]{\widetilde{\theta}^{(#1)}}

\newcommand{\trans}[1]{D_{#1}}
\newcommand{\esttrans}[1]{\widehat{D}_{#1}}

\newcommand{\term}[1]{\mathbb{Y}_{#1}}
\newcommand{\episodecount}[1]{m\left(#1\right)}

\onecolumn 

\newcommand{\edit}[1]{\textcolor{black}{#1}}

\newif\ifarxiv
\arxivtrue
\begin{document}
\title{Optimism-Based Adaptive Regulation of Linear-Quadratic Systems}
\author{Mohamad Kazem~Shirani Faradonbeh,
	Ambuj~Tewari,
	and~George~Michailidis}
\maketitle

\begin{abstract}
\edit{The main challenge for adaptive regulation of linear-quadratic systems is the trade-off between identification and control. An adaptive policy needs to address both the estimation of unknown dynamics parameters (exploration), as well as the regulation of the underlying system (exploitation). To this end, optimism-based methods which bias the identification in favor of optimistic approximations of the true parameter are employed in the literature. A number of asymptotic results have been established, but their finite time counterparts are few, with important restrictions.}

\edit{This study establishes results for the worst-case regret of optimism-based adaptive policies. The presented high probability upper bounds are optimal up to logarithmic factors. The non-asymptotic analysis of this work requires very mild assumptions; (i)~stabilizability of the system's dynamics, and (ii)~limiting the degree of heaviness of the noise distribution. To establish such bounds, certain novel techniques are developed to comprehensively address the probabilistic behavior of dependent random matrices with heavy-tailed distributions.}
\end{abstract}

\section{Introduction} \label{Intro} \label{background}
\IEEEPARstart{A}{daptive} control of Linear-Quadratic (LQ) state space models represents a canonical problem, and is the main focus of this work. Such a model describes the dynamics of the system as follows: starting from the initial state $x(0) \in \R^{p}$, its temporal evolution and cost are determined by
\begin{eqnarray}
x(t+1) &=& A_0x(t)+B_0u(t)+ w(t+1), \label{systemeq1} \\
\instantcost{t}&=& x(t)'Qx(t) + u(t)'Ru(t), \label{systemeq2}
\end{eqnarray}
for $t=0,1,\cdots$. The vector $x(t) \in \R^p$ denotes the output (and state) of the system at time $t$, $u(t) \in \R^r$ represents the control signal, and the stochastic process of the noise sequence is denoted by $\left\{ w(t) \right\}_{t=1}^\infty$. Further, the quadratic function $\instantcost{t}$ corresponds to the instantaneous cost of the system (the transpose of the vector $v$ is denoted by $v'$). The transition matrix $A_0 \in \mathbb{R}^{p \times p}$ and the input matrix $B_0 \in \mathbb{R}^{p \times r}$ which constitute the dynamical parameters of the system are {\em unknown}, while the positive definite matrices of the cost, $Q \in \R^{p \times p}, R \in \R^{r \times r}$ are assumed known. 

The main goal is to adaptively regulate the system in order to minimize its long-term average cost. This canonical problem has been thoroughly studied in the literature and a number of asymptotic results have been established, as briefly summarized next. However, finite time results are scarce and rather incomplete, despite their need in applications (e.g. network systems~\cite{faradonbeh2016optimality}). \edit{Note that the theoretical guarantee for fast stabilization of general linear systems has been recently established~\cite{faradonbeh2018stabilization}, but the existing analysis of the regulation problem of cost minimization leads to a remarkable loss of generality, as will be discussed shortly.}

Since the system dynamics are unknown, a popular adaptive procedure for regulation is based on the principle of {\em Certainty Equivalence} (CE) \cite{bar1974dual}. Alternating between estimation and regulation, CE applies a control feedback {\em as if} the identified parameters $A_0,B_0$ are the true matrices that drive the system's evolution~\cite{lai1986asymptotically,guo1988convergence,lai1991parallel}. However, it has been shown that the CE based strategy can lead to wildly incorrect parameter estimates \cite{lai1982least,becker1985adaptive,kumar1990convergence}, and thus suitable modifications have been introduced in the literature \cite{campi1998adaptive,bittanti2006adaptive}. A popular approach, known as {\em Optimism in the Face of Uncertainty} (OFU) \cite{lai1985asymptotically}, was developed to address the suboptimality of CE. In OFU, after constructing a confidence set for the model parameters, a regulation policy is designed 
based on the most {\em optimistic} parameter in the confidence set \cite{campi1997achieving}. 

The above references establish the asymptotic convergence of the \emph{average} cost to the optimal value. However, non-asymptotic results on the growth rate of {\em regret} (i.e. the accumulative deviation from the optimal cost, see \eqref{regretdeff}) have recently appeared \cite{abbasi2011regret,ibrahimi2012efficient}. These papers provide a near-optimal upper bound for the regret of OFU, under the following rather restrictive conditions: 
\begin{enumerate}
	\item \label{(i)} 
	The dynamics matrices are assumed to be {\em controllable} and {\em observable}. This leads to an excessive complexity in the computation of the adaptive regulator. Further, this assumption restricts the applicability of the analysis since the condition may be violated in many LQ systems.
	%\item \label{(ii)}
	%An additional piece of information regarding the true parameters is {\em a priori} available; namely, Frobenius norm \cite{abbasi2011regret} or identifiability \cite{ibrahimi2012efficient}.
	\item \label{(iii)}
	The {\em operator norm} of the closed-loop matrix is less than one, which excludes a remarkable fraction of systems with stable closed-loop matrices. In fact, a stable matrix can have an arbitrarily large operator norm. Note that condition \ref{(i)} only implies that the largest closed-loop eigenvalue (not the operator norm) is less than one~\cite{bertsekas1995dynamic}.
	\item \label{(iv)}
	The noise distribution satisfies a tail condition such as {\em sub-Gaussianity} \cite{abbasi2011regret} or Gaussianity \cite{ibrahimi2012efficient}. Moreover, the coordinates of the noise vectors are uncorrelated.
\end{enumerate}

\edit{This work aims to address these shortcomings by providing a comprehensive treatment of the problem. We study optimality of OFU policies for an extensive family of LQ systems by establishing upper bounds for the \emph{worst-case} regret, under a minimal set of assumptions.} Namely, we remove condition \ref{(i)} above, and replace the strict condition \ref{(iii)} with \emph{stabilizability}, which is the necessary assumption for the optimal control problem to be well-defined. Further,
the high probability near-optimal upper bound for regret established in this work holds for a class of \emph{heavy-tailed} noise vectors with arbitrary correlation structures, thus significantly
relaxing condition \ref{(iv)}. \edit{To the authors' knowledge, this work is the first addressing the non-asymptotic analysis of the regret of adaptive policies for general LQ systems.}

There are a number of conceptual and technical difficulties one needs to address in order to obtain the results of optimal regulation. First, existing methodology for analyzing adaptive policies~\cite{bittanti2006adaptive,abbasi2011regret,ibrahimi2012efficient} becomes non-applicable beyond condition \ref{(iii)}. One reason is due to the fact that matrix multiplication preserves the operator norm; i.e. the norm of the product is upper bounded by the product of the norms. However, the product of two stable matrices can have eigenvalues of arbitrarily large magnitude. Further, sub-Weibull distributions assumed in this work do not need to have generating functions~\cite{faradonbeh2018finite}. Hence, new tools are required to establish concentration inequalities for random matrices with heavy-tailed probability distributions~\cite{tropp2012user,brown1971martingale}.

\edit{In addition, an adaptive strategy is needed to stabilize the system so that the uncertainty about $A_0,B_0$ does not lead to instability. Adaptive stabilization methods are proposed before, and their finite time performance analysis is provided~\cite{faradonbeh2018stabilization}. First, a coarse approximation of the unknown dynamics parameter is shown to be enough for stabilization. Then, it is established that such approximations can be achieved by employing independent random feedbacks in sufficiently many periods. Nevertheless, for non-asymptotic analysis of the performance of regulation policies, a comprehensive study is not currently available, and is adopted as the focus of this study. In case the operator is concerned with stability issues, the algorithm in the aforementioned reference can be applied a priori to the regulation algorithms we discuss here. }

The remainder of this paper is organized as follows. Section \ref{Optimal Policies} formally defines the problem. Section
\ref{Estimation} addresses the problem of accurate estimation of the closed-loop matrix and includes the analysis of the empirical covariance matrix, as well as a high probability prediction bound. Finally, an optimism-based algorithm for adaptive regulation of the system is presented in Section \ref{Adaptive Control Algorithms}. We show that the regret of {\bf Algorithm \ref{Gregretalgo}} is with high-probability optimal, up to a logarithmic factor.  

%\subsection{Notation}
The following notation is used throughout this paper. For matrix $A \in \C^{p \times q}$, $A'$ is its transpose. When $p=q$, the smallest (respectively largest) eigenvalue of $A$ (in magnitude) is denoted by $\lambda_{\min} (A)$ (respectively $\lambda_{\max}(A)$) and the trace of $A$ is denoted by $\tr{A}$. For $\gamma \in \mathbb{R}, \gamma \geq 1, v \in \mathbb{C}^q$,
the norm of $v$ is $\norm{v}{\gamma} = \left(\sum\limits_{i=1}^{q} \left|v_i\right|^\gamma \right)^{1/\gamma}$. 
Further, when $\gamma=\infty$, the norm is defined according to $\norm{v}{\infty} = \max \limits_{1 \leq i \leq q} |v_i|$. We also use the following notation for the operator norm of matrices. For $\beta, \gamma \in \left[1,\infty\right]$, and $A \in \C^{p \times q}$, define 
$$\Opnorm{A}{\gamma}{\beta} = \sup \limits_{v \in \C^{q} \setminus \{0\}} \frac{\norm{Av}{\beta}}{\norm{v}{\gamma}}.$$
Whenever $\gamma = \beta$, we simply write $\Mnorm{A}{\beta}$. \ifarxiv To denote the dimension of manifold $\mathcal{M}$ over the field $F$, we use $\dimension{F}{\mathcal{M}}$.\fi Finally, the sigma-field generated by random vectors $X_1,\cdots,X_n$ is denoted by $\sigma \left( X_1,\cdots,X_n \right)$. The notation for $\para,\Kmatrix{\para}, \Lmatrix{\para}$, and $\extendedLmatrix{\para}$ are provided in Definition \ref{paradeff}, equations \eqref{ricatti2}, \eqref{ricatti1}, and Definition \ref{Ltildedeff}, respectively. \edit{Finally, $\log$ is employed throughout the paper to refer to the natural logarithm function.}

\section{Problem Formulation} \label{Optimal Policies}
First, we formally discuss the problem of adaptive regulation this work is addressing. Equation \eqref{systemeq1} depicts the dynamics of the system, where $\left\{ w(t) \right\}_{t=1}^\infty$ are independent mean-zero noise vectors with full rank covariance matrix $C$: 
\begin{equation*}
\E {w(t)}=0, \:\: \E {w(t)w(t)'}=C, \:\:\: \eigmin{C}>0.
\end{equation*}
The results established also hold if the noise vectors are martingale difference sequences. The true dynamics are assumed to be stabilizable, as defined below.
\begin{deff}[Stabilizability \cite{bertsekas1995dynamic}] \label{stabilizability}
	$\left[A_0,B_0\right]$ is stabilizable if there is $L \in \R ^{r \times p}$ such that $\eigmax{A_0+B_0L} < 1$. The linear feedback matrix $L$ is called a stabilizer.
\end{deff}
\begin{deff}[Notation $\para$] \label{paradeff}
	We use $\para$ to denote the dynamics parameter $\left[A,B\right]$, where $A$ and $B$ are $p \times p$ and $p \times r$ matrices, respectively. Obviously $\para \in \R^{p \times q}$, for $q=p+r$. \edit{In particular, we frequently refer to $\para_0=\left[A_0,B_0\right]$ throughout the paper.}
\end{deff} 
Here, we consider {\em perfect observations}, i.e. the output of the system corresponds to the state vector itself. Next, an admissible control policy is a mapping $\policy$ which designs the control action according to the dynamics matrix $\para_0$, the cost matrices $Q,R$, and the history of the system; i.e. for all $t \geq 0$, $$u(t) = \policy \left( \para_0 , Q, R , \left\{ x(i) \right\}_{i=0}^t, \left\{ u(j) \right\}_{j=0}^{t-1} \right).$$ An {\em adaptive} policy is ignorant about the parameter $\para_0$. So, $$u(t) = \policy \left( Q, R , \left\{ x(i) \right\}_{i=0}^t, \left\{ u(j) \right\}_{j=0}^{t-1} \right).$$ When applying the policy $\policy$, the resulting instantaneous quadratic cost at time $t$ defined according to \eqref{systemeq2} is denoted by $\instantcost{t}^{(\policy)}$. If there is no superscript, the corresponding policy will be clear from the context. \edit{For arbitrary policy $\policy$, let $\avecost{\para_0}{\policy}$ be the average cost of the system:
\begin{equation*}
\avecost{\para_0}{\policy} = \limsup \limits_{T \to \infty} \frac{1}{T} \sum \limits_{t=1}^T {\instantcost{t}^{(\policy)}}.
\end{equation*}}
Note that the dependence of $\avecost{\para_0}{\policy}$ to the known cost matrices $Q,R$ is suppressed. Then, the optimal average cost is defined by $\optcost{\para_0} = \min \limits_{\policy} \avecost{\para_0}{\policy}$, where the minimum is taken over {\em all} admissible policies. Further, $\policy^\star$ is called an {\em optimal} policy for system $\para$, if satisfying $\avecost{\para}{\policy^\star} = \optcost{\para}$. To find $\policy^\star$ for general $\para \in \R^{p \times q}$, one has to solve a \emph{Riccati} equation. A solution, is a positive semidefinite matrix $\Kmatrix{\para}$ satisfying 
\begin{eqnarray} 
	&&\Kmatrix{\para} = Q + A'\Kmatrix{\para}A \notag \\
	&&- A' \Kmatrix{\para}B \left(B'\Kmatrix{\para}B+R\right)^{-1} B'\Kmatrix{\para}A . \label{ricatti2}
\end{eqnarray}
The following result establishes optimality of the linear feedback provided by $\Kmatrix{\para}$ according to
\begin{eqnarray}
\Lmatrix{\para} &=& -\left(B'\Kmatrix{\para}B+R\right)^{-1} B'\Kmatrix{\para}A. \label{ricatti1}
\end{eqnarray}
\begin{deff}[Policy $\policy^\star$]
	\edit{Henceforth, let $\policy^\star$ denote the linear feedback policy $u(t) = \Lmatrix{\para_0} x(t)$, for all $t \geq 0$.} 
\end{deff}
\begin{lem}[Optimality \cite{faradonbeh2018stabilization}] \label{stabilizable}
	If $\para_0$ is stabilizable, then \eqref{ricatti2} has a unique solution, $\policy^\star$ is optimal, and $\optcost{\para_0}=\tr {\Kmatrix{\para_0}C}$. 
	Conversely, if $\Kmatrix{\para_0}$ is a solution of \eqref{ricatti2}, $\Lmatrix{\para_0}$ is a stabilizer. 
\end{lem}
Note that in the latter case of Lemma \ref{stabilizable}, the existence of a solution $\Kmatrix{\para_0}$ implies that it is unique, $\policy^\star$ is an optimal policy, and $\optcost{\para_0}=\tr{ \Kmatrix{\para_0} C}$.

In order to measure the quality of (adaptive) policy $\policy$, the resulting cost will be compared to the optimal average cost defined above. More precisely, letting $\instantcost{t}^{(\policy)}$ be the resulting instantaneous cost at time $t$, {\em regret} at time $T$ is defined as 
\begin{equation} \label{regretdeff}
\regret{T} = \sum \limits_{t=1}^T \left[ \instantcost{t}^{(\policy)} - \optcost{\para_0} \right].
\end{equation}
The comparison between adaptive control policies is made according to regret. The next result describes the asymptotic distribution of the regret. Lemma \ref{minregret}, which is basically a Central Limit Theorem for $\regret{T}$, states that even when applying optimal policy, the regret $\regret{T}$ scales as $\order{T^{1/2}}$, multiplied by a normal random variable. 
{\begin{lem} \label{minregret}
	Applying $\policy^\star$, \edit{let $D=A_0+B_0\Lmatrix{\para_0}$ be the closed-loop matrix.} Then, $T^{-1/2} \regret{T}$ converges in distribution to $\mathcal{N}\left(0,\sigma^2\right)$ as $T$ grows, where
	\begin{eqnarray*}
	\sigma^2 &=& 4\: \tr{ \Kmatrix{\para_0} C \Kmatrix{\para_0} \sum\limits_{n=1}^\infty D^n C {D'}^n } \\
	&+& \lim\limits_{T\to \infty} T^{-1} \sum\limits_{t=1}^T \var \left[w(t)'\Kmatrix{\para_0} w(t)\right]>0.
	\end{eqnarray*}
\end{lem}}
\edit{The proof of Lemma \ref{minregret} based on an application of the martingale Central Limit Theorem~\cite{brown1971martingale} is deferred to the \ifarxiv appendix\else supplementary materials \cite{faradonbeh2017finite}\fi. In the sequel, we discuss the result of Lemma~\ref{minregret}. In the definition of regret in \eqref{regretdeff}, the cumulative deviation from the optimal average cost can be decomposed into the following two fractions:
\\{\bf (i)} The \emph{probabilistic} fraction contributed by the stochastic evolution of the system and randomness of $\left\{w(t)\right\}_{t=1}^\infty$. 
\\{\bf (ii)} The \emph{statistical} fraction caused by the uncertainty about the dynamics and unknownness of $\para_0$ to the operator. 
\\Lemma \ref{minregret} states that the probabilistic fraction scales with the growth rate $\order{T^{1/2}}$. So, trying to push the statistical fraction of the regret (which is due to the error in learning the unknown dynamics) to have a rate less than $\order{T^{1/2}}$ is actually \emph{unnecessary}. Further, Lemma \ref{minregret} provides a lower bound for the worst-case regret of adaptive policies. Since the optimal policy for minimizing the expected cumulative cost $\sum\limits_{t=0}^T \E{\instantcost{t}}$ converges to $\policy^\star$ as $T$ grows \cite{bertsekas1995dynamic}, the regret of an arbitrary policy can not be smaller than that of $\policy^\star$. On the other hand, the high probability upper bound of a normal distribution is in magnitude at least $\left( - \log \delta \right)^{1/2}$.} Therefore, Lemma \ref{minregret} implies that a high probability regret bound to hold with probability at least $1-\delta$, needs to be at least of the order of magnitude of ${T^{1/2}} \left( - \log \delta \right)^{1/2}$. 
\edit{Note that the above argument does not necessarily imply \emph{impossibility} of the smaller magnitudes for the statistical fraction of the regret\footnote{for example, applying $\policy^\star$, we get $\lim\limits_{T \to \infty}T^{-1/2} \E{\regret{T}}=0$.}. However, since there are information theoretic limits in learning the unknown parameter $\para_0$, statistical regret can not be small. A rigorous derivation of lower bounds for the statistical regret is beyond the scope of this work. Although, later on we will intuitively discuss efficiency of the rate $T^{1/2}$, based on the decomposition being used in the regret analysis of Section \ref{Adaptive Control Algorithms}.}
\begin{deff}[Notation $\extendedLmatrix{\para}$] \label{Ltildedeff}
	For arbitrary stabilizable $\para_1,\para_2$, let $\extendedLmatrix{\para_1}=\begin{bmatrix} I_p \\ \Lmatrix{\para_1} \end{bmatrix}$. So, $\para_2 \extendedLmatrix{\para_1}=A_2+B_2\Lmatrix{\para_1}$.
\end{deff}
%%%%%%%%%%%%%%%%%%%%%%%%%%%%%%%%%%%%%%%%%%%%%%%%%%%%%%%%%%%%%%%%%%%%%%%%%%%%%%%%%%
\section{Closed-Loop Identification} \label{Estimation}
When applying linear feedback $L \in \R^{r \times p}$ to the system, the closed-loop dynamics becomes $x(t+1)=\trans{}x(t)+w(t+1)$, where $\trans{}=A_0+B_0L$. Subsequently, we present bounds for the time length the user can interact with the system in order to have sufficiently many observations for accurate identification of the closed-loop matrix. The next set of results are used later on to construct the confidence sets being used to design the adaptive policy. \edit{Since the focus is on adaptive policies for \emph{regulating} the system, the matrix $\trans{}$ is assumed to be stable.}

\edit{First, we define least-squares estimation for matrix $\trans{}$, as follows. Observing the state vectors $\left\{x(t)\right\}_{t=0}^n$, for an arbitrary matrix $M \in \R^{p \times p}$ consider the sum-of-squares loss function 
\begin{equation*}
	\loss{n}{}{M}=\sum\limits_{t=0}^{n-1} \norm{x(t+1)-M x(t)}{2}^2.
\end{equation*}
Then, the true closed-loop transition matrix $\trans{}$ is estimated by $\esttrans{n} $, which is a minimizer of the above loss: $\loss{n}{}{\esttrans{n}} = \min\limits_{M \in \R^{p \times p}} \loss{n}{}{M}$. Solving for $\esttrans{n}$, one can easily see that it admits the closed form expression
\begin{equation*}
\esttrans{n}= \sum\limits_{t=0}^{n-1} x(t+1)x(t)'V_n^{-1},
\end{equation*}
where $V_n=\sum\limits_{t=0}^{n-1} x(t)x(t)'$ denotes the (invertible) empirical covariance matrix of the state process. Therefore, the behavior of $V_n$ needs to be carefully studied. To this end, one needs to tightly examine the state sequence $\left\{x(t) \right\}_{t=0}^{n}$, which in turn highly depends on both the spectral properties of the transition matrix $\trans{}$, as well as the noise process $\left\{ w(t) \right\}_{t=1}^n$. The former is reflected through the constant $\MJordanconst{}{\trans{}}$, while the latter is indicated by $\noisemax{n}{\delta}$ we shortly define.} 

To proceed, let $\trans{}=P^{-1}\Lambda P$ be the Jordan decomposition of $\trans{}$; i.e. $\Lambda$ is block diagonal, $\Lambda= \diag{\Lambda_1,\cdots, \Lambda_k}$, where for all $i=1,\cdots, k$, $\Lambda_i$ is a Jordan matrix of $\lambda_i$:
\begin{equation*}
\Lambda_i = \begin{bmatrix}
\lambda_i & 1 & 0 & \cdots  & 0 \\
0 & \lambda_i & 1 &  \cdots & 0 \\
\vdots & \vdots & \vdots & \vdots & \vdots \\
0 & 0  & \cdots & 0 & \lambda_i
\end{bmatrix} \in \mathbb{C}^{m_i \times m_i}.
\end{equation*}
\begin{deff} [Constant $\MJordanconst{}{\trans{}}$]  \label{Jordandeff} \label{2Jordandeff}
	Denote the Jordan decomposition described above by $\trans{}=P^{-1} \Lambda P$. Letting 
	\begin{equation*}
	\MJordanconst{t}{\Lambda_i}=\inf\limits_{\rho \geq \left|\lambda_i\right|} t^{m_i-1} \rho^t \sum\limits_{j=0}^{m_i-1} \frac{\rho^{-j}}{j!},
	\end{equation*}
	for $t \geq 1$, define $\MJordanconst{t}{\Lambda}= \max \limits_{1 \leq i \leq k} \MJordanconst{t}{\Lambda_i}$. Then, let $\MJordanconst{0}{\Lambda}=1$, and
	\begin{equation*}
	\MJordanconst{}{\trans{}} = \Opnorm{P^{-1}}{\infty}{2} \Mnorm{P}{\infty} \sum\limits_{t=0}^{\infty} \MJordanconst{t}{\Lambda}.
	\end{equation*}
\end{deff}
{Letting $\overline{\lambda}=\eigmax{\trans{}}$, if $\trans{}$ is \emph{diagonalizable}, then clearly  $\MJordanconst{}{\trans{}} \leq \Opnorm{P^{-1}}{\infty}{2} \Mnorm{P}{\infty} \left(1 - \overline{\lambda}\right)^{-1}$. In general, denoting the dimension of the largest block in the Jordan decomposition of $\trans{}$ by $\mult{}=\max\limits_{1 \leq i \leq k} m_i$, we have $\MJordanconst{t}{\Lambda} \leq t^{\mult{}-1} \overline{\lambda}^t e^{1/\overline{\lambda}}$; i.e.
\begin{eqnarray*}
	\MJordanconst{}{\trans{}} \leq \Opnorm{P^{-1}}{\infty}{2} \Mnorm{P}{\infty} e^{1/\overline{\lambda}}
	\left[\frac{\mult{}-1}{- \log \overline{\lambda}} +  \frac{ \left( \mult{} -1 \right)!}{\left(- \log \overline{\lambda}\right)^{\mult{}}}\right].
\end{eqnarray*}}
Toward studying the effect of the noise vectors on the state process, the following tail condition is assumed.
\begin{assum}[Sub-Weibull distribution \cite{faradonbeh2018finite}] \label{tailcondition} 
	There are positive reals $\tailconst, \tailcoeff$, and $\tailexp$, such that for all $t\geq 1; 1 \leq i \leq p$; and $y >0$,  
	\begin{equation*}
	\PP{\left|w_i(t)\right| > y} \leq \tailconst \: \exp \left(-\frac{y^\tailexp}{\tailcoeff}\right).
	\end{equation*}	
\end{assum}
Clearly, the smaller the exponent $\tailexp$ is, the heavier the tail of $w_i(t)$ will be. Assuming a sub-Weibull distribution for the noise coordinates is more general than the sub-Gaussian (or sub-exponential) assumption routinely made in the literature of non-asymptotic analysis \cite{abbasi2011regret}, where $\tailexp \geq 2$ ($\tailexp \geq 1$). \edit{To gain insight into the basic properties of sub-Weibull distributions, consider the setting $\tailexp<1$. It delivers an extensive family of distributions for which moments of all orders are well-defined, while the moment generating function does not exist. So, it relaxes more restrictive tail conditions to a minimal framework that finite time concentration results can be established. Further, Assumption \ref{tailcondition} encompasses fundamental distributions that sub-Exponential families fail to capture, such as polynomials of Gaussian random variables.} Finally, to obtain analogous results for uniformly bounded noise sequences, it suffices to let $\tailexp \to \infty$ in the subsequently presented materials. 

In order to study magnitudes of the state vectors over time, define: 
\begin{eqnarray}
\noisemax{n}{\delta} &=& \left(\tailcoeff \log \left( \frac{\tailconst np}{\delta}\right)\right)^{1/\tailexp} , \label{noisemaxdeff}\\
\predbound{n}{\delta} &=& \MJordanconst{}{\trans{}} \left(\norm{x(0)}{\infty}  + \noisemax{n}{\delta} \right). \label{predbounddeff}
\end{eqnarray}
Lemma \ref{noisebound} and Lemma \ref{statenorm} show that $\noisemax{n}{\delta},\predbound{n}{\delta}$ are the high probability uniform bounds for the size of the noise and the state vectors. As a matter of fact, $\noisemax{n}{\delta},\predbound{n}{\delta}$ scale as $\log^{1/\tailexp}\left( n/\delta\right)$. Hence, for uniformly bounded noise, both of them are fixed constants. 
Then, recalling that $C$ is the positive definite covariance matrix of the noise vectors, let $\samplesize{\ref{stablemin}}{\epsilon}{\delta}$ be large enough, such that the followings hold for all $n \geq \samplesize{\ref{stablemin}}{\epsilon}{\delta}$:
\begin{eqnarray}
\frac{n}{\noisemax{n}{\delta}^2} &\geq& \frac{18 \eigmax{C}+2\epsilon}{ \epsilon^2} p \log \left(\frac{4p}{\delta}\right) , \label{stablemineq1}\\
\frac{n}{\predbound{n}{\delta}^2 \noisemax{n}{\delta}^2} &\geq& \frac{288}{\epsilon^2} p \Mnorm{D}{2}^2 \log \left(\frac{4p}{\delta}\right) ,\label{stablemineq2} \\
\frac{n}{\predbound{n}{\delta}^2} &\geq& \frac{6}{\epsilon} \left(\Mnorm{D}{2}^2+1\right) .\label{stablemineq3}
\end{eqnarray} 
The following result provides a high probability lower bound for the smallest eigenvalue of $V_{n+1}$. \edit{Essentially, Theorem \ref{stablemin} determines the number of state observations needed to ensure that the excitation is persistent enough to identify the closed-loop matrix \cite{green1985persistence,lai1986extended}.}
\begin{thm}[Empirical covariance] \label{stablemin}
	If $n \geq \samplesize{\ref{stablemin}}{\epsilon}{\delta}$, then
	\begin{equation*}
	\PP{\eigmin{V_{n+1}} < n \left( \eigmin{C}-\epsilon \right)} < 2 \delta.
	\end{equation*}
	Moreover, $\lim\limits_{n \to \infty} n^{-1}V_n=\sum\limits_{i=0}^{\infty} D^i C {D'}^i$.
\end{thm}
\begin{proof}[\bf {Proof}]
	First, for $n \geq 1$, and $0<\delta<1$, define the event:
	\begin{equation} \label{eventdeff}
	\event{W} = \left\{ \max\limits_{1 \leq t \leq n} \norm{w(t)}{\infty} \leq \noisemax{n}{\delta} \right\}.
	\end{equation}
	\edit{We use the following intermediate results, for which the proofs are delegated to the \ifarxiv appendix\else supplement, due to space limitations (also available online \cite{faradonbeh2017finite})\fi.} 
	\begin{lem} \label{noisebound}
		Defining $\event{W}$ according to \eqref{eventdeff}, we have $\PP{\event{W}} \geq 1-\delta$. 
	\end{lem}
	\begin{lem} \label{statenorm}
		The following holds on the event $\event{W}$ in \eqref{eventdeff}:
		\begin{equation*}
		\max\limits_{1 \leq t \leq n} \norm{x(t)}{2} \leq \predbound{n}{\delta}.
		\end{equation*}
	\end{lem}
	\begin{lem} \label{empiricalcov}
		Let the event $\event{W}$ be as \eqref{eventdeff}, and define $C_n=n^{-1}\sum\limits_{i=1}^{n}w(i)w(i)'$. Then, on $\event{W}$ we have $\PP{\eigmax{C_n-C}>\epsilon} \leq \delta$, if 
		\begin{equation} \label{empiricalcoveq1}
		\frac{n}{\noisemax{n}{\delta}^2} \geq \frac{6 \eigmax{C}+2\epsilon}{3 \epsilon^2} p \log \left(\frac{2p}{\delta}\right).
		\end{equation} 
	\end{lem}
	\begin{lem} \label{crossproduct}
		Let $U_n=n^{-1}\sum\limits_{i=0}^{n-1} \big[ Dx(i)w(i+1)' + w(i+1)x(i)'D'\big]$, and define $\event{W}$ by \eqref{eventdeff}. Then, on $\event{W}$ we have $\PP{\eigmax{U_n}>\epsilon} \leq \delta$, if
		\begin{equation} \label{crossproducteq1}
		\frac{n}{\Mnorm{D}{2}^2\noisemax{n}{\delta}^2\predbound{n}{\delta}^2} \geq \frac{32 p}{\epsilon^2} \log \left(\frac{2p}{\delta}\right).
		\end{equation} 
	\end{lem}
	Next, note that $x(t+1)=Dx(t)+w(t+1)$ implies
	\begin{eqnarray*}
		V_{n+1} = x(0)x(0)'+D\sum\limits_{i=0}^{n-1}x(i)x(i)'D' +nU_n+nC_n,
	\end{eqnarray*}
	where $C_n,U_n$ are defined in Lemma \ref{empiricalcov}, and Lemma \ref{crossproduct}. So, we obtain the Lyapunov equation $V_{n+1}=DV_{n+1}D'+nE_n$, for
	\begin{equation*}
	E_n=U_n+C_n+ \frac{D\left(x(0)x(0)' - x(n)x(n)'\right)D'}{n} +\frac{x(0)x(0)'}{n},
	\end{equation*}
	i.e.
	\begin{equation} \label{samplelyapeq}
	V_{n+1} = n \sum\limits_{i=0}^{\infty}D^i E_n {D'}^i.
	\end{equation} 
	Henceforth, suppose that $\event{W}$ holds. According to Lemma \ref{empiricalcov}, \eqref{stablemineq1} implies that
	\begin{equation} \label{stablemineproof1}
	\PP{\eigmax{C_n-C}>\frac{\epsilon}{3}} \leq \frac{\delta}{2}.
	\end{equation}
	In addition, by Lemma \ref{crossproduct}, \eqref{stablemineq2} implies that
	\begin{equation} \label{stablemineproof2}
	\PP{\eigmax{U_n}>\frac{\epsilon}{3}} \leq \frac{\delta}{2}.
	\end{equation} 
	Finally, using Lemma \ref{statenorm}, by \eqref{stablemineq3} we get
	\begin{equation} \label{stablemineproof3}
	\frac{1}{n}\left(\Mnorm{D}{2}^2 +1\right) \left( \norm{x(0)}{2}^2 + \norm{x(n)}{2}^2\right) \leq \frac{\epsilon}{3}.
	\end{equation}
	Putting \eqref{stablemineproof1}, \eqref{stablemineproof2}, and \eqref{stablemineproof3} together, on $\event{W}$, with probability at least $1-\delta$, it holds that $\eigmin{E_n} \geq \eigmin{C}-\epsilon$. Therefore, since \eqref{samplelyapeq} implies that $\eigmin{V_{n+1}} \geq n \eigmin{E_n}$, we get the desired result.\ifarxiv Moreover, since $\eigmax{E_n} \leq \eigmax{C}+ \epsilon$, for $2\epsilon = \eigmin{C}$, with probability at least $1-2 \delta$ we have 
	\begin{eqnarray}
	\eigmax{\frac{1}{n}V_{n+1}} &\leq& \sum\limits_{i=0}^{\infty} \eigmax{D^i E_n {D'}^i}
	\leq \frac{3}{2}\eigmax{C} \MJordanconst{}{D'}^2. \label{samplecoveigmax}
	\end{eqnarray}
	\fi When $n \to \infty$, the conditions hold for arbitrary positive values of $\epsilon, \delta$. Thus, we have $\eigmax{E_n-C} \to 0$, which according to \eqref{samplelyapeq} implies the desired result.
\end{proof}
The following corollary provides a high probability confidence set for $D$, which will be used later in Algorithm \ref{Gregretalgo}. Using the bounds $\noisemax{n}{\delta}, \predbound{n}{\delta}$ introduced in \eqref{noisemaxdeff}, \eqref{predbounddeff}, define the prediction bound $\prediction{n}{\delta}$ according to:
\begin{equation} \label{predictiondeff}
\prediction{n}{\delta} = \frac{16 np}{\left(n-1\right)\eigmin{C}} \predbound{n}{\delta}^2 \noisemax{n}{\delta}^2 \log \left(\frac{2p}{\delta}\right).
\end{equation}
\begin{cor} [Prediction bound] \label{stableprediction}
	Define $\prediction{n}{\delta}$ by \eqref{predictiondeff}. Then, $n \geq \samplesize{\ref{stablemin}}{\eigmin{C}/2}{\delta}+1$ implies that 
	\begin{equation*}
	\PP{\Mnorm{ {V_n}^{1/2} \left(\esttrans{n}-D\right)'}{2}^2 > \prediction{n}{\delta}} \leq 3 \delta.
	\end{equation*} 
\end{cor}
\begin{proof}[\bf Proof]
	\edit{First, since $n \geq \samplesize{\ref{stablemin}}{\eigmin{C}/2}{\delta}+1$, similar to the proof of Theorem \ref{stablemin}, on the event $\event{W}$ defined in \eqref{eventdeff}, with probability at least $1-\delta$, we have $\eigmin{V_n} \geq \eigmin{C} \left(n-1\right)/2$.
	Then, as long as $V_n$ is nonsingular, one can write $\esttrans{n}-D= \left(\sum\limits_{t=0}^{n-1}w(t+1)x(t)'\right) V_n^{-1}$, which yields $\left(\esttrans{n}-D\right) V_n \left(\esttrans{n}-D\right)' = U_n' V_n^{-1}U_n$, where $U_n=\sum\limits_{t=0}^{n-1}x(t)w(t+1)'$. Therefore, 
	\begin{equation} \label{stablepredictioneq1}
	\Mnorm{\left(\esttrans{n}-D\right) V_n \left(\esttrans{n}-D\right)'}{2} \leq  \frac{\Mnorm{U_n}{2}^2}{\eigmin{V_n}}.
	\end{equation}
	To proceed, for arbitrary matrix $H \in \R^{k \times \ell}$, define the dilation 
	\begin{equation*}
	\symmetrizer{H} = \begin{bmatrix} 0_{k \times k} & H \\ H' & 0_{\ell \times \ell} \end{bmatrix} \in \R^{(k+\ell) \times (k+\ell)}.
	\end{equation*}
	A well known fact states that the equality $\Mnorm{H}{2} = \eigmax{\symmetrizer{H}}$ holds \cite{tropp2012user}. So, letting $Z_t=x(t)w(t+1)'$, apply the following random matrix concentration inequality to $X_t=\symmetrizer{Z_t} \in \R^{2p \times 2p}$. 
	\begin{lem} \cite{tropp2012user} \label{MAzuma}
		Let $\left\{ X_i \right\}_{i=1}^n$ be a martingale difference sequence of symmetric $p \times p$ matrices adapted to the filtration $\left\{\mathcal{F}_i \right\}_{i=0}^n$. Assume for fixed symmetric matrices $\left\{ M_i \right\}_{i=1}^n$, all matrices $M_i^2-X_i^2$ are positive semidefinite. Then, letting $\sigma^2 = \eigmax{\sum\limits_{i=1}^{n}M_i^2}$, for all $y \geq 0$ we have
		\begin{equation*}
		\PP{\eigmax{\sum\limits_{i=1}^{n}X_i} \geq y} \leq 2p \!\ \exp \left(-\frac{y^2}{8\sigma^2}\right).
		\end{equation*}
	\end{lem}
	Since
	\begin{equation*}
	{X_t}^2 = \begin{bmatrix} \norm{w(t+1)}{2}^2 x(t)x(t)' & 0_{p\times p} \\ 0_{p \times p} & \norm{x(t)}{2}^2 w(t+1)w(t+1)' \end{bmatrix},
	\end{equation*}
	by Lemma \ref{noisebound} and Lemma \ref{statenorm}, all matrices ${M_t}^2-{X_t}^2$ are positive semidefinite on the event $\event{W}$ defined in \eqref{eventdeff}, with $M_t= \symmetrizer{p^{1/2} \noisemax{n}{\delta} \predbound{n}{\delta} I_p}$. By $\sigma^2= np \noisemax{n}{\delta}^2 \predbound{n}{\delta}^2$, letting $y=8^{1/2} \sigma \log^{1/2}\left(\frac{2p}{\delta}\right)$, Lemma \ref{MAzuma} implies $\PP{\Mnorm{U_n}{2}>y} = \PP{\eigmax{\symmetrizer{U_n}}>y} \leq \delta$. Plugging in \eqref{stablepredictioneq1}, we get the desired result.}
\end{proof}

%%%%%%%%%%%%%%%%%%%%%%%%%%%%%%%%%%%%%%%%%%%%%%%%%%%%%%%%%%%%%%%%%%%%%%%%%%%%%%%%%%

\section{Design of Adaptive Policy} \label{Adaptive Control Algorithms}

In this section, we present an algorithm for adaptive regulation of LQ systems. 
When applying the following algorithm, we assume that a stabilizing set is provided. Construction of such a set with an arbitrary high probability guarantee is addressed in the literature \cite{faradonbeh2018stabilization}. It is established that the proposed adaptive stabilization procedure returns a stabilizing set in finite time. 
%Note that the aforementioned analysis is fairly general such that the restrictive assumptions we discussed in Section \ref{Intro} are not required. 
Nevertheless, if such a set is not available, the operator can apply the proposed method of random linear feedbacks \cite{faradonbeh2018stabilization} in order to stabilize the system before running the following adaptive policy.

In the episodic algorithm below, estimation will be reinforced at the end of every episode. Indeed, the algorithm is based on a sequence of confidence sets, which are constructed according to Corollary \ref{stableprediction}. This sequence will be tightened at the end of every episode; i.e. the provided confidence sets become more and more accurate. According to this sequence, the adaptive linear feedback will be updated after every episode. After explaining the algorithm, we present a high probability regret bound.

First, we provide a high level explanation of the algorithm. Starting with the stabilizing set $\paraspace{0}$, we select a parameter $\optpara{1} \in \paraspace{0}$ based on OFU principle; i.e. $\optpara{1}$ is a minimizer of the optimal average cost over the corresponding confidence set (see \eqref{algo2eq1}). 

Then, assuming $\optpara{1}$ is the true parameter the system evolves according to, during the first episode the algorithm applies the optimal linear feedback $\Lmatrix{\optpara{1}}$. Once the observations during the first episode are collected, they are used to improve the accuracy of the high probability confidence set. Therefore, $\paraspace{0}$ is tightened to $\paraspace{1}$, and the second episode starts by selecting $\optpara{2}$, iterating the above procedure, and so on. The lengths of the episodes will be increasing, to make every confidence set significantly more accurate than all previous ones.

The intuition behind proficiency of the OFU principle is as follows. Applying a linear feedback $L$, the closed-loop transition matrix is $A_0+B_0L=\para_0 \widetilde{L}$, where $\widetilde{L}= \left[I_p, L'\right]'$. Importantly, the observed sequence of state vectors accurately identifies the closed-loop matrix $\para_0 \widetilde{L}$. However, an accurate estimation of $\para_0 \widetilde{L}$ does not lead to that of $\para_0$. Therefore, $\para_0$ is not guaranteed to be effectively approximable, regardless of the accuracy in the approximation of $\para_0 \widetilde{L}$.

Nevertheless, one has to focus on finding accurate approximations of the feedback matrix $\Lmatrix{\para_0}$, in order to design an effective adaptive policy for minimizing the average cost. Specifically, as long as $\para_1$ is available satisfying $\Lmatrix{\para_1}= \Lmatrix{\para_0}$, one can apply an optimal linear feedback $\Lmatrix{\para_1}$, no matter how large $\Mnorm{\para_1 - \para_0}{2}$ is. In general, estimation of such a $\para_1$ is not possible. Yet, an optimistic approximation in addition to exact knowledge of the closed-loop dynamics lead to an optimal linear feedback, thanks to the OFU principle. 
\begin{lem} \label{OFU}
	If $\optcost{\para_1} \leq \optcost{\para_0}$, and $\para_1 \extendedLmatrix{\para_1}=\para_0 \extendedLmatrix{\para_1}$, then $\Lmatrix{\para_1}$ is optimal for the system $\para_0$: $\Lmatrix{\para_0}=\Lmatrix{\para_1}$.
\end{lem} 
In other words, applying linear feedback $\Lmatrix{\para_{1}}$ which is designed according to an optimistically selected parameter $\para_1$, as long as the closed-loop matrix $\para_0 \extendedLmatrix{\para_1}$ is exactly identified, the optimal linear feedback is automatically provided. Recall that the lengths of the episodes are growing so that the estimation of the closed-loop matrix becomes more precise at the end of every episode. Thus, the approximation $\para_1 \extendedLmatrix{\para_1} \approx \para_0 \extendedLmatrix{\para_1}$ is becoming more and more accurate. Rigorous analysis of the discussion above, leads to the high probability near-optimal regret bound of Theorem \ref{Gregret}. 

\edit{Algorithm \ref{Gregretalgo} takes the stabilizing set $\paraspace{0}$, the failure probability $\firstcoeff\delta$, and the reinforcement rate $\rrate >1$ as inputs.} Indeed, $\paraspace{0}$ is a {bounded} stabilizing set such that for every $\para \in \paraspace{0}$, the system will be stable if the optimal linear feedback of $\para$ is applied; i.e. $\eigmax{\para_0 \extendedLmatrix{\para}} < 1$. As mentioned before, an algorithmic procedure to obtain a bounded stabilizing set in finite time is available in the literature \cite{faradonbeh2018stabilization}. Furthermore, $\firstcoeff\delta >0$ is the highest probability that Algorithm \ref{Gregretalgo} fails to adaptively regulate the system such that the regret will be nearly optimal (see Theorem \ref{Gregret}). The reinforcement rate $\rrate $ determines the growth rate of the lengths of the time intervals (episodes) a specific feedback is applied until being updated (see \eqref{algo2eq2}). 
\begin{algorithm}
	\caption{{\bf: Adaptive Regulation}} \label{Gregretalgo}
	\begin{algorithmic}
		\State {\bf Inputs}: $\paraspace{0} \subset \R^{p \times q}$, $\firstcoeff\delta>0$, $\rrate >1$
		\State Let $\tau_0=0$
		\For{$i=1,2,\cdots$}
		\State Define $\optpara{i}$, $\tau_i$ according to \eqref{algo2eq1}, \eqref{algo2eq2}, respectively
		\While{$t < \tau_i$}
		\State Apply control feedback $u(t)=\Lmatrix{\optpara{i}} x(t)$
		\EndWhile
		\State Find the estimate $\widehat{D}^{(i)}$ given in \eqref{algo2eq4}
		\State Using $V^{(i)}$ in \eqref{algo2eq5}, construct $\tempparaspace{i}$ according to \eqref{algo2eq6}
		\State Update $\paraspace{i}$ by \eqref{algo2eq7}
		\EndFor
	\end{algorithmic}
\end{algorithm}

The algorithm provides an adaptive policy as follows. For $i=1,2,\cdots$, at the beginning of the $i$-the episode, we apply linear feedback $u(t)=\Lmatrix{\optpara{i}} x(t)$, where
\begin{equation} \label{algo2eq1}
\optpara{i} \in \arg\min \limits_{\para \in \paraspace{i-1}} \optcost{\para}.
\end{equation}
Indeed, based on OFU principle, at the beginning of every episode, {\em the most optimistic} parameter amongst all we are uncertain about is being selected. The length of episode $i$, which is the time period we apply the adaptive control policy $u(t)=\Lmatrix{\optpara{i}} x(t)$, is designed according to the following equation. Letting $\tau_0=0$, we update the control policy at the end of episode $i$, i.e. at the time $t=\tau_i$, defined according to
\begin{equation} \label{algo2eq2}
\tau_i = \tau_{i-1} +  \rrate^{i/q}  \samplesize{\ref{stablemin}}{\frac{\eigmin{C}}{2}}{\frac{\delta}{i^2}}+\rrate^{i/q} ,
\end{equation}
where $\samplesize{\ref{stablemin}}{\cdot}{\cdot}$ is defined by \eqref{stablemineq1}, \eqref{stablemineq2}, and \eqref{stablemineq3}. After the $i$-th episode, we estimate the closed-loop transition matrix $\para_0 \extendedLmatrix{\optpara{i}}$ by the following least-squares estimator:
\begin{eqnarray}
\widehat{D}^{(i)} = \arg\min\limits_{M \in \R^{p\times p}} \sum\limits_{t=\tau_{i-1}}^{\tau_i-1} \norm{x(t+1)-Mx(t)}{2}^2. \label{algo2eq4}
\end{eqnarray}
Letting $V^{(i)}$ be the empirical covariance matrix of episode $i$,
\begin{equation}  \label{algo2eq5}
V^{(i)} = \sum\limits_{t=\lceil \tau_{i-1} \rceil }^{ \lceil \tau_i \rceil-1} x(t)x(t)',
\end{equation}
define the high probability confidence set
\ifarxiv
\begin{eqnarray} 
\tempparaspace{i} = \Bigg\{ \para \in \R^{p \times q} &:& \Mnorm{ {V^{(i)}}^{\frac{1}{2}} \left(\para \extendedLmatrix{\optpara{i}}-\widehat{D}^{(i)}\right)' }{2}^2
\leq \prediction{\tau_i-\tau_{i-1}}{\frac{\delta}{i^2}} \Bigg\}, \label{algo2eq6}
\end{eqnarray}
\else
\begin{eqnarray} 
\tempparaspace{i} = \Bigg\{ \para \in \R^{p \times q} : \Mnorm{ {V^{(i)}}^{1/2}  \left(\para \extendedLmatrix{\optpara{i}}-\widehat{D}^{(i)}\right)' }{2}^2 \notag \\
\leq \prediction{\tau_i-\tau_{i-1}}{\frac{\delta}{i^2}} \Bigg\}, \label{algo2eq6}
\end{eqnarray}\fi 
where $\prediction{n}{\delta}$ is defined in \eqref{predictiondeff}. Note that according to Corollary \ref{stableprediction}, $\PP{\para_0 \in \tempparaspace{i}} \geq 1-3 \delta i^{-2}$. Then, at the end of episode $i$, the confidence set $\paraspace{i-1}$ will be updated to
\begin{equation} \label{algo2eq7}
\paraspace{i} = \paraspace{i-1} \cap \tempparaspace{i},
\end{equation}
and episode $i+1$ starts, finding $\optpara{i+1}$ by \eqref{algo2eq1}, and then iterating all steps described above. 
\begin{remark} \label{Gregretalgoremark}
	The choice of $\optpara{i}$ does not need to be as extreme as \eqref{algo2eq1} \cite{abbasi2011regret}. In fact, it suffices to satisfy $\optcost{\optpara{i}} \leq \left(\tau_i-\tau_{i-1}\right)^{-1/2}+ \inf \limits_{\para \in \paraspace{i-1}} \optcost{\para}$.
\end{remark}
The following result states that performance of the above adaptive control algorithm is optimal, apart from a logarithmic factor. {Theorem \ref{Gregret} also provides the effect of the degree of heaviness of the noise distribution (denoted by $\tailexp$ in Assumption \ref{tailcondition}) on the regret.} Compared to $\order{\cdot}$, the notation $\orderlog{\cdot}$ used below, hides the logarithmic factors. 
\begin{thm}[Regret bound] \label{Gregret}
	For bounded $\paraspace{0}$, with probability at least $1-\firstcoeff\delta$, the regret of Algorithm \ref{Gregretalgo} satisfies:
	\begin{equation*}
	\regret{T} \leq \orderlog{T^{1/2} \left( - \log \delta \right)^{1/2+2/\tailexp}}.
	\end{equation*}
\end{thm}
\begin{proof}[\bf {Proof}]
	The stabilizing set $\paraspace{0}$ is bounded:
	\begin{equation} \label{boundedT}
	\ssconstant_{1} = \sup\limits_{\para \in \paraspace{0}} \Mnorm{\para'}{2} < \infty.
	\end{equation}
	Suppose that for $t=1,2,\cdots$, the parameter $\para_t$ is being used to design the adaptive linear feedback $u(t)=\Lmatrix{\para_t}x(t)$. So, during every episode, $\para_t$ does not change, and for $\tau_{i-1} \leq t < \tau_i$ we have $\para_t= \optpara{i}$. 
	
	Letting $\mathcal{F}_t = \sigma \left(w(1), \cdots, w(t)\right)$, the infinite horizon dynamic programming equations \cite{bertsekas1995dynamic} are
	\ifarxiv
	\begin{eqnarray*}
		\optcost{\para_t} + x(t)' \Kmatrix{\para_t} x(t) = x(t)'Q x(t) + u(t)'Ru(t)
		+ \E{y(t+1)' \Kmatrix{\para_t} y(t+1) \Big | \mathcal{F}_t},
	\end{eqnarray*}
	\else 
	\begin{eqnarray*}
		\optcost{\para_t} &+& x(t)' \Kmatrix{\para_t} x(t) = x(t)'Q x(t) + u(t)'Ru(t) \\
		&+& \E{y(t+1)' \Kmatrix{\para_t} y(t+1) \Big | \mathcal{F}_t},
	\end{eqnarray*}\fi
	where $u(t)=\Lmatrix{\para_t} x(t)$, and
	\begin{equation} \label{Gproofeq1}
	y(t+1)= A_t x(t) + B_t u(t) + w(t+1) = \para_t \extendedLmatrix{\para_t} x(t) + w(t+1)
	\end{equation}
	describes the desired dynamics of the system. Note that since the true evolution of the system is governed by $\para_0$, the next state is 
	\begin{equation} \label{Gproofeq2}
	x(t+1)= A_0 x(t) + B_0 u(t) + w(t+1) = \para_0 \extendedLmatrix{\para_t} x(t) + w(t+1).
	\end{equation}
	Substituting \eqref{Gproofeq1}, and \eqref{Gproofeq2} in the dynamic programming equation, and using \eqref{systemeq2} for the instantaneous cost $\instantcost{t}$, we have
	\ifarxiv 
	\begin{eqnarray*}
		\optcost{\para_t} + x(t)' \Kmatrix{\para_t} x(t)
		&=& \instantcost{t} + \E{ w(t+1)' \Kmatrix{\para_t} w(t+1) \Big | \mathcal{F}_t }
		+ x(t) \extendedLmatrix{\para_t}' \para_t' \Kmatrix{\para_t} \para_t \extendedLmatrix{\para_t} x(t) \\
		&=& \instantcost{t} + \E{x(t+1)' \Kmatrix{\para_t} x(t+1) \Big | \mathcal{F}_t}
		+ x(t) \extendedLmatrix{\para_t}' \left[\para_t' \Kmatrix{\para_t} \para_t - \para_0' \Kmatrix{\para_t} \para_0\right] \extendedLmatrix{\para_t} x(t).
	\end{eqnarray*}
	\else 
	\begin{eqnarray*}
		&& \optcost{\para_t} + x(t)' \Kmatrix{\para_t} x(t) \\
		&=& \instantcost{t} + \E{ w(t+1)' \Kmatrix{\para_t} w(t+1) \Big | \mathcal{F}_t } \\
		&+& x(t) \extendedLmatrix{\para_t}' \para_t' \Kmatrix{\para_t} \para_t \extendedLmatrix{\para_t} x(t) \\
		&=& \instantcost{t} + \E{x(t+1)' \Kmatrix{\para_t} x(t+1) \Big | \mathcal{F}_t} \\
		&+& x(t) \extendedLmatrix{\para_t}' \left[\para_t' \Kmatrix{\para_t} \para_t - \para_0' \Kmatrix{\para_t} \para_0\right] \extendedLmatrix{\para_t} x(t).
	\end{eqnarray*}\fi  
	Adding up the terms for $t=1,\cdots, T$, we obtain:
	\begin{equation} \label{Gproofeq3}
	\regret{T}= \sum\limits_{t=1}^{T} \left[\instantcost{t} - \optcost{\para_0}\right] =  \term{1} + \term{2} + \term{3} + \term{4},
	\end{equation}
	where \ifarxiv
	\begin{eqnarray}
	\term{1} &=& \sum\limits_{t=1}^{T} \left[\optcost{\para_t}-\optcost{\para_0}\right], \label{term1}\\
	\term{2} &=& \sum\limits_{t=1}^{T} \Big( x(t)' \Kmatrix{\para_t} x(t) - \E{x(t+1)' \Kmatrix{\para_{t+1}} x(t+1) \Big| \mathcal{F}_t} \Big), \label{term2} \\
	\term{3} &=& \sum\limits_{t=1}^{T} \E{ x(t+1)' \left(\Kmatrix{\para_{t+1}} - \Kmatrix{\para_t}\right) x(t+1) \Big| \mathcal{F}_t} , \label{term3}\\
	\term{4} &=& \sum\limits_{t=1}^{T} x(t)' \extendedLmatrix{\para_t}' \Big[ \para_0' \Kmatrix{\para_t} \para_0 
	- \para_t' \Kmatrix{\para_t} \para_t \Big] \extendedLmatrix{\para_t} x(t). \:\:\:\:\:\:\: \label{term4}
	\end{eqnarray}\else
	the expressions for $\term{1}, \term{2}, \term{3}, \term{4}$ are defined in \eqref{term1}-\eqref{term4}.
	\begin{table*}
		\begin{eqnarray}
		\term{1} &=& \sum\limits_{t=1}^{T} \left[\optcost{\para_t}-\optcost{\para_0}\right], \label{term1}\\
		\term{2} &=& \sum\limits_{t=1}^{T} \Big( x(t)' \Kmatrix{\para_t} x(t) - \E{x(t+1)' \Kmatrix{\para_{t+1}} x(t+1) \Big| \mathcal{F}_t} \Big), \label{term2} \\
		\term{3} &=& \sum\limits_{t=1}^{T} \E{ x(t+1)' \left(\Kmatrix{\para_{t+1}} - \Kmatrix{\para_t}\right) x(t+1) \Big| \mathcal{F}_t} , \label{term3}\\
		\term{4} &=& \sum\limits_{t=1}^{T} x(t)' \extendedLmatrix{\para_t}' \Big[ \para_0' \Kmatrix{\para_t} \para_0 
		- \para_t' \Kmatrix{\para_t} \para_t \Big] \extendedLmatrix{\para_t} x(t). \:\:\:\:\:\:\: \label{term4}
		\end{eqnarray}
	\end{table*}
	\fi 
	Let $\episodecount{T}$ be the number of episodes considered until time $T$. Thus,
	\begin{equation*}
	\tau_{\episodecount{T}} \leq T < \tau_{\episodecount{T}+1}.
	\end{equation*}
	Now, letting $n_i = \lfloor \tau_i-\tau_{i-1} \rfloor$ be the length of episode $i$, define the following events
	\begin{eqnarray*}
		\event{G} &=& \bigcap\limits_{i=1}^\infty \left\{ \max\limits_{\tau_{i-1} \leq t < \tau_i} \norm{w(t)}{\infty} \leq \noisemax{n_i}{\frac{\delta}{i^2}} \right\} , \\
		\event{H} &=& \bigcap\limits_{i=1}^\infty \left\{ \para_0 \in \paraspace{i} \right\}.
	\end{eqnarray*}
	According to Corollary \ref{stableprediction}, 
	\begin{equation} \label{Gproofeq4}
	\PP{\event{G} \cap \event{H}} \geq 1-\sum\limits_{i=1}^{\infty} \frac{3 \delta}{i^2} \geq 1- 5 \delta. 
	\end{equation}
	For all $i=1,2,\cdots$, as long as $\para_0 \in \paraspace{i-1}$, according to \eqref{algo2eq1} we have $\optcost{\optpara{i}} \leq \optcost{\para_0}$; i.e. $\optcost{\para_t}-\optcost{\para_0} \leq 0$. Therefore, on $\event{G} \cap \event{H}$ we have
	\begin{equation} \label{Gprooflem1}
	\term{1} \leq 0.
	\end{equation}
	\edit{To conclude the proof of the theorem, we leverage some auxiliary results. The proofs of the following lemmas are deferred to \ifarxiv appendix.\else supplement due to space limitations (available online \cite{faradonbeh2017finite}).\fi}
	\begin{lem}[Bounding $\term{2}$] \label{Gprooflem2}
		On $\event{G} \cap \event{H}$, the following holds with probability at least $1-\delta$:
		\begin{equation*}
		\term{2} \leq \ssconstant_2 + (8T)^{1/2} \ssconstant_3 \Big(\log \left(T \episodecount{T}\right)\Big)^{2/\tailexp} \left( -\log \delta \right)^{1/2+ 2/\tailexp}, 
		\end{equation*}
		where $\ssconstant_2, \ssconstant_3 < \infty$ are fixed constants.
	\end{lem}
	\begin{lem}[Bounding $\term{3}$] \label{Gprooflem3}
		On $\event{G} \cap \event{H}$, we have
		\begin{equation*}
		\term{3} \leq \ssconstant_3 \Big(\log \left(T \episodecount{T}\right)\Big)^{2/\tailexp} \left( -\log \delta \right)^{2/\tailexp} \episodecount{T},
		\end{equation*}
		where $\ssconstant_3$ is the same as Lemma \ref{Gprooflem2}.
	\end{lem}
	\begin{lem}[Bounding $\term{4}$] \label{Gprooflem5}
		On the event $\event{G} \cap \event{H}$, it holds that
		\begin{equation*}
		\term{4} \leq \ssconstant_{4} \episodecount{T}^{3/2} \prediction{T}{\frac{\delta}{\episodecount{T}^2}}^{1/2} T^{1/2},
		\end{equation*}
		for some fixed constant $\ssconstant_{4} < \infty$.
	\end{lem}
	\begin{lem}[Bounding $\episodecount{T}$] \label{Gprooflem6}
		On the event $\event{G} \cap \event{H}$ the following holds:
		\begin{equation*}
		\episodecount{T} \leq \frac{q}{\log \rrate} \log \left( \frac{T \left(\rrate^{1/q}-1\right)}{\tau_1}+1\right).
		\end{equation*}
	\end{lem}
	
	Finally, the definition of $\prediction{n}{\delta}$ in \eqref{predictiondeff} yields
	\begin{equation*}
	\prediction{n}{\delta} = \order{\left( \log n \right)^{4/\tailexp} \left( -\log \delta \right)^{1+ 4/\tailexp}}.
	\end{equation*}
	Therefore, plugging \eqref{Gprooflem1}, and the results of Lemmas  \ref{Gprooflem2}, \ref{Gprooflem3}, \ref{Gprooflem5}, and \ref{Gprooflem6} into \eqref{Gproofeq3}, we get $\regret{T} \leq \orderlog{T^{1/2} \left( -\log \delta \right)^{1/2+ 2/\tailexp}}$, with probability at least $1-\delta$ on $\event{G} \cap \event{H}$. Hence, according to \eqref{Gproofeq4}, the failure probability is at most $\firstcoeff\delta$, which completes the proof.
\end{proof}
\edit{To conclude this section, we briefly discuss the behavior of  the statistical regret introduced in the discussion after Lemma~\ref{minregret}. For this purpose, we use the regret decomposition of \eqref{Gproofeq3} into the terms $\term{1}, \cdots, \term{4}$ being defined in \eqref{term1} - \eqref{term4}. According to Lemma \ref{Gprooflem3}, $\term{3}$ scales logarithmically with $T$. Further, the martingale $\term{2}$ is bounded in expectation; i.e. $\limsup\limits_{T \to \infty} \E{\term{2}}<\infty$. Hence, one can approximately study the behavior of the statistical regret by addressing ${\term{1}}, {\term{4}}$. First, note that the expression $\para_0' \Kmatrix{\para_t} \para_0 - \para_t' \Kmatrix{\para_t} \para_t$ in \eqref{term4} can be substituted by $\left(\para_0 + \para_t\right)' \Kmatrix{\para_t} \left(\para_0 - \para_t\right)$. Since $\Kmatrix{\para_0}$ is positive definite~\cite{faradonbeh2018stabilization}, the magnitude of ${\term{4}}$ is approximately as large as $\sum\limits_{t=1}^T \Mnorm{\para_t - \para_0}{2}$. A similar argument applies to ${\term{1}}$ in the sense that the decay rate of  $\optcost{\para_t}-\optcost{\para_0}$ heavily relies on the error of learning $\para_0$ through $\para_t$. Then, the learning accuracy at time $t$ is at best of the order~$t^{-1/2}$~\cite{lai1986asymptotically}. Hence, the statistical regret an adaptive policy needs to incur is at least $\order{T^{1/2}}$, because of lack of knowledge about the true parameter . Converting this lower bound sketch into a rigorous proof is beyond the scope of this work, and is left as an interesting problem for future studies.}

\section{Conclusion} \label{future}
This work investigated adaptive regulation schemes for linear dynamical systems with quadratic costs, focusing on finite time analysis for regret.
Using the OFU principle, we established non-asymptotic optimality results under the mild condition of stabilizability and also assuming a fairly general heavy-tailed noise distributions.

There are a number of interesting extensions of the current work. First, generalizing the non-asymptotic analysis of efficiency to {\em imperfect} observations of the state vector is a topic of future investigation. Another interesting direction is to specify the sufficient and necessary conditions for the true dynamics which lead to optimality of {\em Certainty Equivalence}. In addition, re-examining the problem for large scale {\em network} systems where the transition matrix can be sparse is also of interest.
\ifarxiv \else
\section*{Acknowledgment}
The authors would like to thank Professor Tze Leung Lai for helpful discussions, and the reviewers for their constructive comments.\fi 

\appendices
\section{Proofs in Sections \ref{Optimal Policies}, \ref{Estimation}}
\begin{proof}[\bf \edit{Proof of Lemma \ref{minregret}}]
	When applying the linear feedback $\Lmatrix{\para_0}$, the closed-loop transition matrix will be $D=\para_0 \extendedLmatrix{\para_0}=A_0+B_0 \Lmatrix{\para_0}$. Letting $P=Q+\Lmatrix{\para_0}'R\Lmatrix{\para_0}$, we have the followings. First,
	\begin{eqnarray*}
		&& \regret{T-1}+x(0)'Px(0)-\optcost{\para_0} = \sum\limits_{t=0}^{T-1}x(t)'Px(t) - T\optcost{\para_0}  
		= \tr{PV_T}- T\optcost{\para_0},
		\end{eqnarray*}
		where $V_T= \sum_{t=0}^{T} x(t)x(t)'$. Second, $x(t+1)=Dx(t)+w(t+1)$ implies $V_{T}=DV_{T}D'+E_T$, where
			\begin{eqnarray*}
			E_T &=& U_T+C_T+x(0)x(0)'
			+ D\left(x(0)x(0)' - x(T-1)x(T-1)'\right)D' , \\
		U_T &=& \sum\limits_{t=0}^{T-2} \left[Dx(t)w(t+1)'+w(t+1)x(t)'D'\right] ,\\
		C_T &=& \sum\limits_{t=1}^{T-1} w(t) w(t)'.
		\end{eqnarray*}
		Third, stability of the system yields $\lim\limits_{T \to \infty} {T^{-1/2}}\norm{x(T-1)}{2}^2=0$, almost surely. Putting all the above together, in addition to $\optcost{\para_0}=\tr{\Kmatrix{\para_0}C}$, and the well known fact \cite{bertsekas1995dynamic}
		\begin{equation} \label{ricatti3}
		\Kmatrix{\para_0}-D' \Kmatrix{\para_0} D = Q + \Lmatrix{\para_0}'R \Lmatrix{\para_0},
		\end{equation}
		we get
		\begin{eqnarray*}
		\lim\limits_{T \to \infty} \frac{1}{T^{1/2}} \regret{T}
		&=& \lim\limits_{T \to \infty} \frac{1}{T^{1/2}} \tr{P \sum\limits_{n=0}^{\infty}D^nE_T {D'}^n -T C\sum\limits_{n=0}^{\infty}{D'}^n P {D}^n}
		= \sum\limits_{n=0}^{\infty}\tr{{D'}^n P D^n \lim\limits_{T \to \infty}  \frac{U_T+C_T-TC}{T^{1/2}}}.
		\end{eqnarray*}
		Letting $\sigma_T^2 = \var \left( T^{-1/2} \regret{T}\right)$, according to Lindeberg's Central Limit Theorem \cite{brown1971martingale}, the last expression above converges in distribution to $\mathcal{N}\left(0,\sigma^2\right)$, where $\sigma^2=\lim\limits_{T\to \infty} \sigma_T^2$. In order to find $\sigma^2$, using $\lim\limits_{T \to \infty} T^{-1/2} \E{\regret{T}}=0$, $\lim\limits_{t \to \infty} \E{x(t)}=0$, and $\lim\limits_{t \to \infty} \cov \left( x(t) \right) = \sum\limits_{n=0}^\infty D^n C {D'}^n$, we obtain
		\begin{eqnarray*}
			\lim\limits_{T\to \infty} T^{-1} \E{\regret{T}^2} &=& \lim\limits_{T\to \infty} T^{-1} \E{ \Bigg( \sum\limits_{t=1}^T \big[w(t)'\Kmatrix{\para_0} w(t) - \optcost{\para_0} + 2 w(t)' \Kmatrix{\para_0} D x(t-1)\big]  \Bigg)^2} \\
		&=& \lim\limits_{T\to \infty} T^{-1} \sum\limits_{t=1}^T \E{ \Big( w(t)'\Kmatrix{\para_0} w(t) - \optcost{\para_0} \right. 
		+ \left. 2 w(t)' \Kmatrix{\para_0} D x(t-1) \Big)^2} \\
		&=& \lim\limits_{T\to \infty} T^{-1} \sum\limits_{t=1}^T \left(\E{ \left[w(t)'\Kmatrix{\para_0} w(t) - \optcost{\para_0}\right]^2} \right. 
		+ \left. 4 \E{ \left[w(t)' \Kmatrix{\para_0} D x(t-1)\right]^2}\right) \\
		&=& 4 \tr{ \Kmatrix{\para_0} C \Kmatrix{\para_0} \sum\limits_{n=1}^\infty D^n C {D'}^n }
		+ \lim\limits_{T\to \infty} T^{-1} \sum\limits_{t=1}^T \var \left[w(t)'\Kmatrix{\para_0} w(t)\right].
		\end{eqnarray*}
\end{proof}
		
\begin{proof}[\bf Proof of Lemma \ref{noisebound}]
	First, note that for all $y>0 ; i=1,\cdots,p ; t=1,\cdots, n$, by Assumption \ref{tailcondition} we have
	\begin{eqnarray*}
		\PP{\left|w_i(t)\right| > \noisemax{n}{\delta}} &\leq& \tailconst \exp \left(-\frac{\noisemax{n}{\delta}^\tailexp}{\tailcoeff}\right) = \tailconst \exp \left(-\frac{\tailcoeff \log \frac{ \tailconst np}{\delta}}{\tailcoeff}\right) = \frac{\delta}{np} . 
	\end{eqnarray*}
	So, using a union bound we get
	\begin{equation*}
	\PP{\event{W}^c} \leq \sum\limits_{t=1}^{n} \sum\limits_{i=1}^{p} \PP{\left|w_i(t)\right| > \noisemax{n}{\delta}} \leq \delta.
	\end{equation*}
\end{proof}

\begin{proof}[\bf Proof of Lemma \ref{statenorm}]	
		First, the behavior of $\Mnorm{\Lambda}{\infty}$ is determined by the blocks of $\Lambda$. In fact, letting $\Lambda=\diag{\Lambda_1, \cdots, \Lambda_k}$, simply the definition of the operator norm $\Mnorm{\cdot}{\infty}$ implies $$\Mnorm{\Lambda}{\infty} \leq \max \limits_{1 \leq i \leq k} \Mnorm{\Lambda_i}{\infty}.$$
		Then, to find an upper bound for the operator norm of an exponent of an arbitrary block, such as 
		\begin{equation*}
		\Lambda_i=\begin{bmatrix} \lambda & 1 & 0 & \cdots & 0\\
		0 & \lambda & 1 & \cdots & 0\\
		\vdots & \vdots & \vdots & \vdots & \vdots \\
		0 & \cdots & 0 & 0 & \lambda\end{bmatrix} \in \C^{m \times m},
		\end{equation*}
		we show that
		\begin{equation} \label{statenormeq1}
		\Mnorm{\Lambda_i^t}{\infty} \leq t^{m-1} \left|\lambda\right|^t \sum\limits_{j=0}^{m-1} \frac{\left|\lambda\right|^{-j}}{j!}.
		\end{equation}
		For this purpose, note that for $k=0,1, \cdots$,
			\begin{equation*}
			\Lambda_i^k = \begin{bmatrix}
			\lambda^k & {k \choose 1} \lambda^{k-1}  & \cdots & {k \choose m-1} \lambda^{k-m+1} \\
			0 & \lambda^k & \cdots & {k \choose m-2} \lambda^{k-m+2} \\
			\vdots & \vdots & \vdots & \vdots \\
			0 & \cdots & 0 & \lambda^k
			\end{bmatrix},
			\end{equation*}
		and let $v \in \C^m$ be such that $\norm{v}{\infty}=1$. For $\ell=1,\cdots,m$, the $\ell$-th coordinate of $\Lambda_i^tv$ is $\sum\limits_{j=0}^{m-\ell} {t \choose j} \lambda^{t-j}v_{j+\ell+1}$, which, because of ${t \choose j} \leq \frac{t^j}{j!}$, is at most the right hand side of \eqref{statenormeq1}. Therefore, because of $\Lambda^t=\diag{\Lambda_1^t, \cdots, \Lambda_k^t}$, we have $\Mnorm{\Lambda^t}{\infty} \leq \MJordanconst{t}{\Lambda}$. Now, by $x(t)=D^t x(0)+ \sum\limits_{i=1}^{t}D^{t-i}w(i)$, by Lemma \ref{noisebound}, on the event $\event{W}$ we have 
		\begin{eqnarray*}
				\norm{x(t)}{2} &=& \norm{P^{-1}\Lambda^tP x(0)+ \sum\limits_{i=1}^{t}P^{-1}\Lambda^{t-i}Pw(i)}{2} \\
				&\leq& \Opnorm{P^{-1}}{\infty}{2} \left(\Mnorm{\Lambda^t}{\infty} \norm{Px(0)}{\infty}+ \sum\limits_{i=1}^{t}\norm{\Lambda^{t-i}Pw(i)}{\infty}\right) \\
				&\leq& \predbound{n}{\delta}.
		\end{eqnarray*}
\end{proof}

\begin{proof}[\bf Proof of Lemma \ref{empiricalcov}]
	In this proof, we use the following Matrix Bernstein inequality \cite{tropp2012user}:
	\begin{lem} \label{MBernstein}
		Let $X_i \in \R^{p \times p}, i=1,\cdots, n$ be a sequence of independent symmetric random matrices. Assume for all $i=1,\cdots, n$, we have $\E{X_i}=0$ and $\eigmax{X_i} \leq \ssconstant$. Then, for all $y \geq 0$ we have
		\begin{equation*}
		\PP{\eigmax{\sum\limits_{i=1}^{n}X_i} \geq y} \leq 2p \!\ \exp \left(-\frac{3y^2}{6\sigma^2 + 2 \ssconstant y}\right),
		\end{equation*}
		where $\sigma^2 = \eigmax{\sum\limits_{i=1}^{n}\E{X_i^2}}$.
	\end{lem}
	Letting $X_i=w(i)w(i)'-C$, and $\ssconstant = p \noisemax{n}{\delta}^2$, clearly $\E{X_i}=0$. Further, we have
	\begin{eqnarray*}
		\sigma^2 &=& \eigmax{\sum\limits_{i=1}^{n}\E{X_i^2}} \\
		&\leq& \sum\limits_{i=1}^{n}\eigmax{\E{\norm{w(i)}{2}^2 w(i)w(i)'}-C^2} \leq n \ssconstant \eigmax{C}, 
	\end{eqnarray*}
	since on $\event{W}$, the inequality $\max\limits_{1 \leq i \leq n}\norm{w(i)}{2}^2 \leq \ssconstant $ holds. Therefore, by \eqref{empiricalcoveq1} we have
	\begin{eqnarray*}
		\PP{\eigmax{C_n-C}>\epsilon} &=& \PP{\eigmax{\sum\limits_{i=1}^{n}X_i}>n\epsilon} \\
		&\leq& 2p \!\ \exp \left(- \frac{3n\epsilon^2}{6 \ssconstant \eigmax{C}+2\ssconstant \epsilon}\right) \leq \delta.
	\end{eqnarray*}
\end{proof}

\begin{proof}[\bf Proof of Lemma \ref{crossproduct}]
	In this proof, we use Matrix Azuma inequality of Lemma \ref{MAzuma}.
%	\begin{lem} \label{MAzuma}
%		Let $X_i \in \R^{p \times p}, i=1,\cdots, n$ be a martingale difference sequence of symmetric matrices, i.e. for some filter $\{\mathcal{F}_i\}_{i=0}^n$, $X_i$ is $\mathcal{F}_i$-measurable and $\E{X_{i+1} | \mathcal{F}_i}=0$. Assume for fixed symmetric matrices $M_i, i=1,\cdots, n$, all matrices $M_i^2-X_i^2$ are positive semidefinite. Then, for all $y \geq 0$ we have
%		\begin{equation*}
%		\PP{\eigmax{\sum\limits_{i=1}^{n}X_i} \geq y} \leq 2p \!\ \exp \left(-\frac{y^2}{8\sigma^2}\right),
%		\end{equation*}
%		where $\sigma^2 = \eigmax{\sum\limits_{i=1}^{n}M_i^2}$.
%	\end{lem}
	Letting 
	\begin{eqnarray*}
		X_i &=& Dx(i-1)w(i)'+w(i)x(i-1)'D', \\
		\mathcal{F}_i &=& \sigma \left(w(1), \cdots,w(i)\right), \\
		M_i &=& 2 p^{1/2} \noisemax{n}{\delta}\predbound{n}{\delta} \Mnorm{D}{2} I_p ,
	\end{eqnarray*}
	clearly, $\E{X_{i+1}|\mathcal{F}_i}=0$. Further, $M_i^2-X_i^2$ is positive semidefinite, since by Lemma \ref{noisebound}, and Lemma \ref{statenorm}, on $\event{W}$ we have 
	\begin{eqnarray*}
		\max\limits_{1 \leq i \leq n} \norm{w(i)}{2} &\leq& p^{1/2} \noisemax{n}{\delta} , \\
		\max\limits_{0 \leq i \leq n-1} \norm{x(i)}{2} &\leq& \predbound{n}{\delta}.
	\end{eqnarray*}
	Therefore, $\sigma^2 = 4np \Mnorm{D}{2}^2 \noisemax{n}{\delta}^2 \predbound{n}{\delta}^2$, and by \eqref{crossproducteq1} we have
	\begin{eqnarray*}
		\PP{\eigmax{U_n}>\epsilon} &=& \PP{\eigmax{\sum\limits_{i=1}^{n}X_i}>n\epsilon} \\
		&\leq& 2p \!\ \exp \left(- \frac{n\epsilon^2}{32p \Mnorm{D}{2}^2 \noisemax{n}{\delta}^2 \predbound{n}{\delta}^2}\right) \leq \delta.
	\end{eqnarray*}
\end{proof}

\section{Proofs in Section \ref{Adaptive Control Algorithms}}
\begin{proof}[\bf Proof of Lemma \ref{OFU}]
	Suppose that $\optcost{\para_1} \leq \optcost{\para_0}$ and $D=\para_1 \extendedLmatrix{\para_1}= \para_0 \extendedLmatrix{\para_1}$. Applying the linear feedback 
	$$\policy_1: u(t)=\Lmatrix{\para_1} x(t), t=0,1,\cdots$$ 
	to a system evolving according to the dynamics parameter $\para_0$, the closed-loop matrix will be $x(t+1)=Dx(t)+w(t+1)$. Letting $P= Q+ \Lmatrix{\para_1}' R \Lmatrix{\para_1}$, we have
	\begin{eqnarray*}
		\E{\instantcost{t}} &=& \E{x(t)' P x(t)} \\
		&=& \E{x(t-1)' D' P D x(t-1)} + \E{w(t)' P w(t)} \\
		&=& \cdots = x(0)' {D'}^t P {D}^t x(0) + \sum\limits_{i=1}^{t} \E{w(i)' {D'}^{t-i} P {D}^{t-i} w(i)}.
	\end{eqnarray*}
	Note that due to the stabilizability of $\para_0$, the inequality $\optcost{\para_1} \leq \optcost{\para_0}$ implies that $\para_1$ is also stabilizable, and hence by Lemma \ref{stabilizable} we have $\eigmax{D}<1$. Thus,
	\begin{equation} \label{OFUproofeq1}
	\lim\limits_{t \to \infty} x(0)' {D'}^t P {D}^t x(0) =0.
	\end{equation} 
	Furthermore, by $\E{w(i)w(i)'}=C$, the second term is $\tr{C \sum\limits_{i=0}^{t-1} {D'}^{i} P {D}^{i}}$. Therefore, using \eqref{OFUproofeq1} we get
	\begin{equation*}
	\lim\limits_{t \to \infty} \E{\instantcost{t}} = \tr{C \sum\limits_{i=0}^{\infty} {D'}^{i} P {D}^{i}}.
	\end{equation*}
	The above convergence holds for the Cesaro mean of the sequence $\left\{ \E{\instantcost{t}} \right\}_{t=1}^\infty$ as well, i.e. the expected average cost is
	\begin{equation*} \label{OFUproofeq3}
	\avecost{\para_0}{\policy_1} = \tr{C \sum\limits_{i=0}^{\infty} {D'}^{i} P {D}^{i}}.
	\end{equation*}
	Similarly, since $u(t)=\Lmatrix{\para_1} x(t)$ is optimal for a system of dynamics parameter $\para_1$,
	\begin{eqnarray*}
		\optcost{\para_0} \geq \optcost{\para_1} = \tr{C \sum\limits_{i=0}^{\infty} {D'}^{i} P {D}^{i}} 
		= \avecost{\para_0}{\policy_1} \geq \optcost{\para_0},
	\end{eqnarray*} 
	i.e. the linear feedback $\Lmatrix{\para_1}$ is an optimal policy for a system of dynamics parameter $\para_0$, which is the desired result. 
\end{proof}

\begin{proof}[\bf Proof of Lemma \ref{Gprooflem2}]
		First, write
		\begin{eqnarray*}
			\term{2} &=& \sum\limits_{t=1}^{T} x(t)' \Kmatrix{\para_t} x(t) - \sum\limits_{t=2}^{T+1} \E{x(t)' \Kmatrix{\para_{t}} x(t) \Big| \mathcal{F}_{t-1}}  \\
			&=& \E{x(1)' \Kmatrix{\para_{1}} x(1)} - \E{x(T+1)' \Kmatrix{\para_{T+1}} x(T+1) \Big| \mathcal{F}_T} \\
			&+& \sum\limits_{t=1}^{T} \left(x(t)' \Kmatrix{\para_t} x(t) - \E{x(t)' \Kmatrix{\para_{t}} x(t) \Big| \mathcal{F}_{t-1}}\right).
		\end{eqnarray*}
		Then, letting
		\begin{equation} \label{Gprooflem2eq1}
		\ssconstant_{0}= \sup\limits_{1 \leq i \leq \infty} \Mnorm{\Kmatrix{\optpara{i}}}{2},
		\end{equation}
		note that the above sequence we are taking supremum on, is bounded because for positive definite matrix $C$, on $\event{H}$ the OFU principle of \eqref{algo2eq1} implies
		\begin{equation*}
		\optcost{\para_0} \geq \optcost{\optpara{i}} = \tr{\Kmatrix{\optpara{i}}C},
		\end{equation*}
		hence, 
		\begin{eqnarray*}
			\optcost{\para_0} &\geq& \tr{C^{1/2} \Kmatrix{\optpara{i}} C^{1/2}} \\
			&\geq& \eigmax{C^{1/2} \Kmatrix{\optpara{i}} C^{1/2}} \\
			&=& \sup\limits_{v \neq 0} \frac{v' \Kmatrix{\optpara{i}} v}{\norm{v}{2}^2} \frac{\norm{v}{2}^2}{v' C^{-1}v} \\
			&\geq& \eigmin{C} \eigmax{\Kmatrix{\optpara{i}}},
		\end{eqnarray*}
		i.e.
		\begin{equation*}
		\ssconstant_{0} \leq \frac{\optcost{\para_0}}{\eigmin{C}} < \infty.
		\end{equation*}
		To proceed, using the boundedness of $\paraspace{0}$,
		\begin{equation} \label{Gprooflem2eq2}
		\E{x(1)' \Kmatrix{\para_{1}} x(1)} = x(0) \extendedLmatrix{\para_1}' \para_1' \Kmatrix{\para_1} \para_1 \extendedLmatrix{\para_1} x(0) + \tr{\Kmatrix{\para_1}C} \leq \ssconstant_2,
		\end{equation}
		for some $\ssconstant_2 < \infty$. Defining the stable closed-loop matrices $D_i= \para_0 \extendedLmatrix{\optpara{i}}, i=1,\cdots, \episodecount{T}$, similar to Lemma \ref{statenorm}, one can simply show that on the event $\event{G}$, for constant $\MJordanconst{}{D_1 , \cdots, D_{\episodecount{T}}}< \infty$, it holds that
		\begin{equation} \label{Gproofstatenorm}
		\max\limits_{1 \leq t \leq T} \norm{x(t)}{2} \leq \MJordanconst{}{D_1 , \cdots, D_{\episodecount{T}}} \noisemax{T}{\frac{\delta}{\episodecount{T}^2}},
		\end{equation}
		where the fact $$\max\limits_{1 \leq i \leq \episodecount{T}} \noisemax{n_i}{\frac{\delta}{i^2}} \leq \noisemax{T}{\frac{\delta}{\episodecount{T}^2}}$$ is used above. Therefore, for martingale difference sequence
		\begin{equation*}
		\left\{ X_t \right\}_{t=1}^T = \left\{ x(t)' \Kmatrix{\para_t} x(t) - \E{x(t)' \Kmatrix{\para_{t}} x(t) \Big| \mathcal{F}_{t-1}} \right\}_{t=1}^T,
		\end{equation*}
		on $\event{G}$ we have 
		\begin{eqnarray*}
			\left| X_t \right| &\leq& 2 \Mnorm{\Kmatrix{\para_t}}{2} \norm{x(t)}{2}^2 \leq 2 \ssconstant_{0} \MJordanconst{}{D_1 , \cdots, D_{\episodecount{T}}}^2 \noisemax{T}{\frac{\delta}{\episodecount{T}^2}}^2
			\leq \ssconstant_3 \Big(\log \left(T \episodecount{T}\right)\Big)^{2/\tailexp} \left( -\log \delta \right)^{2/\tailexp},  
		\end{eqnarray*}
		for some $\ssconstant_3 < \infty$. Letting
		\begin{eqnarray*}
			\sigma^2 &=& \ssconstant_3^2 T \Big(\log \left(T \episodecount{T}\right)\Big)^{4/\tailexp} \left( -\log \delta \right)^{4/\tailexp} \geq \sum\limits_{t=1}^{T} X_t^2 , \\
			y &=& (8T)^{1/2} \ssconstant_3 \Big(\log \left(T \episodecount{T}\right)\Big)^{2/\tailexp} \left( -\log \delta \right)^{1/2+ 2/\tailexp},
		\end{eqnarray*}
		apply Lemma \ref{MAzuma} to get
		\begin{equation*}
		\PP{\sum\limits_{t=1}^{T} X_t > y} \leq \exp \left( -\frac{y^2}{8 \sigma^2}\right) \leq \delta,
		\end{equation*}
		which in addition to \eqref{Gprooflem2eq2} implies the desired result, because of 
		$$\E{x(T+1)' \Kmatrix{\para_{T+1}} x(T+1) \Big| \mathcal{F}_T} \geq 0.$$
\end{proof}

\begin{proof}[\bf Proof of Lemma \ref{Gprooflem3}]
		Note that as long as both of $t$ and $t+1$ are in episode $i$, we have $$\para_t=\para_{t+1}= \optpara{i}.$$
		Thus, using \eqref{Gprooflem2eq1}, and \eqref{Gproofstatenorm}, on $\event{G} \cap \event{H}$ we have
		\begin{eqnarray*}
			\term{3} &=& \sum\limits_{i=1}^{\episodecount{T}-1} \E{ x(\lceil \tau_i \rceil)' \left(\Kmatrix{\optpara{i+1}} - \Kmatrix{\optpara{i}}\right) x(\lceil \tau_i \rceil) \Big| \mathcal{F}_t} , \\
			&\leq& \episodecount{T} \max\limits_{1 \leq i \leq \episodecount{T}-1} \Mnorm{\Kmatrix{\optpara{i+1}}}{2} \norm{x(\lceil \tau_i \rceil)}{2}^2 \\
			&\leq& \ssconstant_{0} \episodecount{T} \MJordanconst{}{D_1 , \cdots, D_{\episodecount{T}}}^2 \noisemax{T}{\frac{\delta}{\episodecount{T}^2}}^2\\
			&\leq& \ssconstant_3 \Big(\log \left(T \episodecount{T}\right)\Big)^{2/\tailexp} \left( -\log \delta \right)^{2/\tailexp} \episodecount{T}.
		\end{eqnarray*}
\end{proof}
	
\begin{lem} [Lipschitz continuity] \label{LipschitzK}
		Suppose that $\para_1 \in \R^{p \times q}$ is stabilizable. Then, there are constants $0 <\epsilon_1, \Lipschitz{K}<\infty$, such that for an arbitrary stabilizable $\para_2 \in \R^{p \times q}$ satisfying $\Mnorm{\para_2 - \para_1}{2}<\epsilon_1$, the following holds: 
		\begin{equation*}
		\Mnorm{\Kmatrix{\para_1}-\Kmatrix{\para_2}}{2} \leq \Lipschitz{K} \Mnorm{\para_1 - \para_2}{2}.
		\end{equation*}
		Note that according to Lemma \ref{stabilizable} and \eqref{ricatti1}, we obtain 
		\begin{eqnarray*}
			\Mnorm{\Lmatrix{\para_2}-\Lmatrix{\para_1}}{2} &\leq& \Lipschitz{L} \Mnorm{\para_2-\para_1}{2}, \\
			\left|\optcost{\para_2}-\optcost{\para_1}\right| &\leq& \Lipschitz{\cost} \Mnorm{\para_2-\para_1}{2},
		\end{eqnarray*}
		for $\Lipschitz{L}, \Lipschitz{\cost}<\infty$.
\end{lem}
	
\begin{proof}[\bf Proof of Lemma \ref{LipschitzK}]
		First, let $D_1,D_2 \in \R^{p \times p}$ be stable, and $P \in \R^{p \times p}$ be a positive semidefinite matrix. For $i=1,2$, define $F_i = \sum\limits_{n=0}^{\infty} D_i'^n P D_i^n$. We show that 
		\begin{equation} \label{LipschitzKeq1}
		\Mnorm{F_1-F_2}{2} \leq \Lipschitz{F} \Mnorm{D_1 - D_2}{2},
		\end{equation}
		for some $\Lipschitz{F}<\infty$. For $n=1,2,\cdots$, we have
		\begin{eqnarray*}
			\Mnorm{D_2^n - D_1^n}{2} &=& \Mnorm{\left(D_1+ D_2-D_1\right)^n - D_1^n}{2} \\
			&\leq& \sum\limits_{a_0+\sum\limits_{j=1}^{m}\left(a_j+b_i\right)=n, \sum\limits_{j=0}^{m}a_j<n} \Mnorm{D_1^{a_0}\prod\limits_{j=1}^{m} \left(D_2-D_1\right)^{b_j}D_1^{a_j}}{2} \\
			&\leq& \sum\limits_{a_0+\sum\limits_{j=1}^{m}\left(a_j+b_i\right)=n, \sum\limits_{j=0}^{m}a_j<n} \Mnorm{D_1}{2}^{\sum\limits_{j=0}^{m} a_j} \Mnorm{D_2-D_1}{2}^{\sum\limits_{j=1}^{m} b_j} \\
			&=& \sum\limits_{\ell=1}^n {n \choose \ell} \Mnorm{D_1}{2}^{n-\ell} \Mnorm{D_1-D_2}{2}^\ell \\
			&\leq& \frac{\left(\Mnorm{D_1}{2} + \Mnorm{D_1-D_2}{2}\right)^n}{\Mnorm{D_1}{2}}n \Mnorm{D_1-D_2}{2}.
		\end{eqnarray*}
		Then, there is $k<\infty$, such that
		\begin{equation*}
		\max \left\{ \Mnorm{D_1'^k}{2}, \Mnorm{D_1^k}{2}, \Mnorm{D_2'^k}{2}, \Mnorm{D_2^k}{2} \right\} \leq 1-2\rho, 
		\end{equation*}
		for some $\rho >0$. Define $$E_i=D_i^k, P_i=\sum\limits_{n=0}^{k-1} D_i'^n P D_i^n.$$
		Noting that $$\Mnorm{D_2'-D_1'}{2} \leq \Lipschitz{0} \Mnorm{D_2-D_1}{2},$$
		we have
		\begin{eqnarray*}
			\Mnorm{E_2-E_1}{2} &\leq& \frac{\left(\Mnorm{D_1}{2} + \Mnorm{D_2-D_1}{2}\right)^k}{\Mnorm{D_1}{2}} k \Mnorm{D_2-D_1}{2} = \Lipschitz{E}\Mnorm{D_2-D_1}{2}, \\
			\Mnorm{E_2'-E_1'}{2} &\leq& \frac{\left(\Mnorm{D_1'}{2} + \Mnorm{D_2'-D_1'}{2}\right)^k}{\Mnorm{D_1'}{2}} k \Mnorm{D_2'-D_1'}{2} = \Lipschitz{E'} \Mnorm{D_2-D_1}{2}, \\
			\Mnorm{P_2-P_1}{2} &\leq& \sum\limits_{n=1}^{k-1} \left[\Mnorm{D_2'^n P \left(D_2^n-D_1^n\right)}{2} + \Mnorm{\left(D_2'^n-D_1'^n\right) P D_1^n}{2}\right] \\
			&\leq& \sum\limits_{n=1}^{k-1} \left[ \Mnorm{D_2'^n}{2} \Mnorm{D_2^n-D_1^n}{2} + \Mnorm{D_1^n}{2} \Mnorm{D_2'^n-D_1'^n}{2}\right] \Mnorm{P}{2} n \\
			&\leq& \Lipschitz{P} \Mnorm{D_2-D_1}{2}. 	
		\end{eqnarray*}
		Suppose that $\Mnorm{D_2-D_1}{2}$ is small enough to satisfy 
		\begin{equation*}
		\max \left\{\Mnorm{E_2-E_1}{2} , \Mnorm{E_2'-E_1'}{2}\right\} \leq \rho.
		\end{equation*}
		Since $\Mnorm{E_1}{2}+\Mnorm{E_1-E_2}{2} \leq 1-\rho$, $\Mnorm{E_1'}{2}+\Mnorm{E_1'-E_2'}{2} \leq 1-\rho$, and $F_i=\sum\limits_{n=0}^{\infty} E_i'^n P_i E_i^n$, similar to the upper bound above for $\Mnorm{D_2^n - D_1^n}{2}$ we have
		\begin{eqnarray*}
			\Mnorm{E_2^n-E_1^n}{2} &\leq& \frac{\left(\Mnorm{E_1}{2}+\Mnorm{E_2-E_1}{2}\right)^n}{\Mnorm{E_1}{2}} n \Mnorm{E_2-E_1}{2} \\
			&\leq& \frac{\Lipschitz{E}}{\Mnorm{E_1}{2}} \Mnorm{D_2-D_1}{2} n \left(1-\rho\right)^n , \\
			\Mnorm{E_2'^n-E_1'^n}{2} &\leq& \frac{\left(\Mnorm{E_1'}{2}+\Mnorm{E_2'-E_1'}{2}\right)^n}{\Mnorm{E_1'}{2}} n \Mnorm{E_2'-E_1'}{2} \\
			&\leq& \frac{\Lipschitz{E'}}{\Mnorm{E_1'}{2}} \Mnorm{D_2-D_1}{2} n \left(1-\rho\right)^n .
		\end{eqnarray*}
		Thus,
		\begin{eqnarray*}
			\Mnorm{F_2-F_1}{2} &\leq& \sum\limits_{n=0}^{\infty} \left[\Mnorm{E_2'^n P_2 \left(E_2^n-E_1^n\right)}{2} + \Mnorm{\left(E_2'^n-E_1'^n\right) P_2 E_1^n}{2} + \Mnorm{E_1'^n \left(P_2-P_1\right)E_1^n}{2}\right] \\
			&\leq& \sum\limits_{n=0}^{\infty} \left[ \frac{\Mnorm{P_2}{2}\Lipschitz{E}}{\Mnorm{E_1}{2}} n + \frac{\Mnorm{P_2}{2} \Lipschitz{E'}}{\Mnorm{E_1'}{2}} n + \Lipschitz{P} \right] \left(1-\rho\right)^{2n} \Mnorm{D_2-D_1}{2} \\
			&=& \Lipschitz{F} \Mnorm{D_2-D_1}{2},
		\end{eqnarray*}
		i.e. \eqref{LipschitzKeq1} holds. Next, to prove the desired inequality, consider two systems $(1),(2)$, with cumulative costs $\cost=\sum\limits_{t=0}^{\infty}\instantcost{t}$, where for $i=1,2$, System ($i$) evolves according to $x(t+1)=A_ix(t)+B_iu(t)$, and both systems share the initial state $x(0)=x_0$, where $\norm{x_0}{2}=1$. Denoting the optimal accumulative cost of System ($i$) by $\cost^{(i)}$, we have $\cost^{(i)}=x_0'\Kmatrix{\para_i}x_0$ \cite{bertsekas1995dynamic}. Let $\epsilon_1$ be sufficiently small such that $\eigmax{\para_1 \extendedLmatrix{\para_2}} \leq 1-\tilde{\rho}, \eigmax{\para_2 \extendedLmatrix{\para_1}} \leq 1-\tilde{\rho}$, for some $\tilde{\rho}>0$. Then, apply control policy $u(t)=\Lmatrix{\para_2}x(t)$ to both systems. The closed-loop matrices $D_i=\para_i \extendedLmatrix{\para_2}$ are stable, and
		\begin{equation} \label{LipschitzKeq2}
		\Mnorm{D_1-D_2}{2} \leq \Mnorm{\extendedLmatrix{\para_2}}{2} \Mnorm{\para_1-\para_2}{2}.
		\end{equation}
		Letting $P=Q+\Lmatrix{\para_2}'R\Lmatrix{\para_2}$, the accumulative cost of System ($i$) is $x_0'F_ix_0$, where $F_i = \sum\limits_{n=0}^{\infty} D_i'^n P D_i^n$. The linear feedback $\Lmatrix{\para_2}$ is an optimal policy for System (2), i.e. $x_0'\Kmatrix{\para_2}x_0=x_0'F_2x_0$, and $\cost^{(1)}$ is the minimum accumulative cost for System (1), i.e. $x_0'\Kmatrix{\para_1}x_0 \leq x_0'F_1x_0$. Therefore, whenever $\cost^{(1)} \geq \cost^{(2)}$, \eqref{LipschitzKeq1}, \eqref{LipschitzKeq2} imply that 
		\begin{equation} \label{LipschitzKeq3}
		0 \leq \cost^{(1)} - \cost^{(2)} = x_0'\Kmatrix{\para_1}x_0 - x_0'\Kmatrix{\para_2}x_0 \leq x_0'\left(F_1-F_2\right)x_0 \leq \Mnorm{F_1-F_2}{2} \leq \Lipschitz{F} \Mnorm{\extendedLmatrix{\para_2}}{2} \Mnorm{\para_1-\para_2}{2}.
		\end{equation}
		Otherwise, if $\cost^{(1)} \leq \cost^{(2)}$, applying $u(t)=\Lmatrix{\para_1}x(t)$ to both systems, $E_i=\para_i \extendedLmatrix{\para_1}$, $i=1,2$ are stable, and
		\begin{equation} \label{LipschitzKeq4}
		\Mnorm{E_1-E_2}{2} \leq \Mnorm{\extendedLmatrix{\para_1}}{2} \Mnorm{\para_1-\para_2}{2}.
		\end{equation}
		Furthermore, the accumulative cost of System ($i$) is $x_0'G_ix_0$, where $G_i = \sum\limits_{n=0}^{\infty} E_i'^n P E_i^n$. Since $\Lmatrix{\para_1}$ is optimal for System (1), $x_0'\Kmatrix{\para_1}x_0=x_0'G_1x_0$, and $x_0'\Kmatrix{\para_2}x_0 \leq x_0'G_2x_0$. Therefore, \eqref{LipschitzKeq1}, \eqref{LipschitzKeq4} lead to 
		\begin{equation} \label{LipschitzKeq5}
		0 \leq \cost^{(2)} - \cost^{(1)} = x_0'\Kmatrix{\para_2}x_0 - x_0'\Kmatrix{\para_1}x_0 \leq x_0'\left(G_2-G_1\right)x_0 \leq \Mnorm{G_2-G_1}{2} \leq \Lipschitz{G} \Mnorm{\extendedLmatrix{\para_1}}{2} \Mnorm{\para_1-\para_2}{2}.
		\end{equation}
		Since $x_0$ is arbitrary, \eqref{LipschitzKeq3}, \eqref{LipschitzKeq5} yield the desired result.
	\end{proof}
	\begin{lem} \label{Wproofauxiliary}
	Let $\left\{ M^{(i)}\right\}_{i=1}^\infty$ be a sequence of $p \times p$ matrices. Whenever $\tau_{i-1} \leq t < \tau_i$, let $M_t=M^{(i)}$. Define
	\begin{equation*}
	\term{} = \sum\limits_{t=1}^{T} \norm{M_tx(t)}{2}^2.
	\end{equation*}
	On $\event{G} \cap \event{H}$, it holds that
	\begin{equation*}
	\term{} \leq \ssconstant_5 \sum\limits_{i=1}^{\episodecount{T}} n_i \Mnorm{M^{(i)}}{2}^2, 
	\end{equation*}
	for some constant $\ssconstant_5 < \infty$.
\end{lem}
\begin{proof}[\bf Proof of Lemma \ref{Wproofauxiliary}]
	Letting $D_i = \para_0 \extendedLmatrix{\optpara{i}}$ be the stable closed-loop matrix during episode $i$, and
	\begin{equation*}
	V^{(i)}= \sum\limits_{t=\tau_{i-1}}^{\tau_i-1} x(t)x(t)',	
	\end{equation*}
	be the empirical covariance matrix of episode $i$, according to \eqref{samplecoveigmax}, on $ \event{G} \cap \event{H}$ we have \begin{equation*}
	\eigmax{V^{(i)}} \leq \frac{3}{2} \eigmax{C} \MJordanconst{}{D_i'}^2 n_i.
	\end{equation*}
	Letting 
	\begin{equation*}
	\ssconstant_5 = \frac{3}{2} p^2 \eigmax{C} \:\:\sup\limits_{i \geq 1} \MJordanconst{}{D_i'}^2 < \infty,
	\end{equation*} 
	we have
	\begin{eqnarray*}
		\term{} &=& \sum\limits_{t=1}^{T} x(t)' M_t' M_tx(t)
		\leq \sum\limits_{i=1}^{\episodecount{T}} \tr{{M^{(i)}} V^{(i)} {M^{(i)}}'} \\
		&\leq& \sum\limits_{i=1}^{\episodecount{T}} p \eigmax{V^{(i)}} \Mnorm{{M^{(i)}}'}{2}^2
		\leq \frac{\ssconstant_5}{p} \sum\limits_{i=1}^{\episodecount{T}} n_i \eigmax{{M^{(i)}}{M^{(i)}}'} \\
		&\leq& \frac{\ssconstant_5}{p} \sum\limits_{i=1}^{\episodecount{T}} n_i \tr{{M^{(i)}}'{M^{(i)}}} \leq \ssconstant_5 \sum\limits_{i=1}^{\episodecount{T}} n_i \Mnorm{M^{(i)}}{2}^2 .
	\end{eqnarray*}
\end{proof}
	
	\begin{lem} \label{Gprooflem4}
		Letting $U_0=I_q$, for $i=1,2,\cdots$ define the symmetric $q \times q$ matrix $U_i$ as
		\begin{equation*}
		U_i = \extendedLmatrix{\optpara{i}} V^{(i)} \extendedLmatrix{\optpara{i}}' = \extendedLmatrix{\optpara{i}} \sum\limits_{t=\tau_{i-1}}^{\tau_i} x(t)x(t)' \extendedLmatrix{\optpara{i}}',
		\end{equation*}
		and for arbitrary nonzero $\para \in \R^{p \times q}$, let the real-valued sequence $\left\{ s_j \left( \para \right) \right\}_{j=1}^\infty$ be
		\begin{equation*}
		s_j \left( \para \right) = \frac{\Mnorm{ \para \sum\limits_{i=0}^{j} U_i \para'}{2}}{\Mnorm{ \para \sum\limits_{i=0}^{j-1} U_i \para'}{2}}.
		\end{equation*}
		Note that $s_j \left( \para \right)$ does not depend on the magnitude of $\para$. The Cesaro mean of the sequence $\left\{ s_j \left( \para \right) \right\}_{j=1}^\infty$ is bounded; i.e. for some constant $\ssconstant_{6}$, on $\event{G} \cap \event{H}$ we have
		\begin{equation*}
		\sup\limits_{n \geq 1} \frac{1}{n} \sum\limits_{j=1}^{n} s_j \left( \para \right) \leq \ssconstant_{6}.
		\end{equation*}
	\end{lem}
	\begin{proof}[\bf Proof of Lemma \ref{Gprooflem4}]
		First, applying the second part of Theorem \ref{stablemin}, we have
		\begin{equation} \label{Gprooflem4eq2}
		\lim\limits_{i \to \infty} \frac{1}{n_i} V^{(i)} = \lim\limits_{i \to \infty} \sum\limits_{\ell=0}^{\infty} {D_i}^\ell C {D_i'}^\ell,
		\end{equation}
		where $D_i = \para_0 \extendedLmatrix{\optpara{i}}$ is the stable closed-loop transition matrix during episode $i$. 
		
		Then, the sequence $\left\{ \Lmatrix{\optpara{i}} \right\}_{i=1}^\infty$ converges as follows. According to \eqref{boundedT}, it is bounded. So, divergence of this bounded sequence implies convergence of two subsequences to distinct limits. Let $L_\infty$ be the limit point of a subsequence. According to \eqref{algo2eq7}, $\left\{ \paraspace{i} \right\}_{i=0}^\infty$ is strictly decreasing: $\paraspace{i+1} \subsetneqq \paraspace{i}$. Further, using Theorem \ref{stablemin} and Corollary \ref{stableprediction}, since $\lim\limits_{i \to \infty} \frac{\prediction{n_i}{\delta i^{-2}}}{n_i}=0$ we have 
		\begin{equation*}
		0= \lim\limits_{i \to \infty} \left(\optpara{i}-\para_0\right) \extendedLmatrix{\optpara{i-1}} = \lim\limits_{i \to \infty} \left(\optpara{i}-\para_0\right) \begin{bmatrix}
		I_p \\ L_\infty
		\end{bmatrix}.
		\end{equation*} 
		So, $L_\infty$ is a stationary point in the sense that for some $\para_\infty \in \bigcap\limits_{i=0}^\infty \paraspace{i}$, we have
		\begin{equation} \label{Gprooflem4eq1}
		A_\infty+B_\infty L_\infty = A_0+B_0L_\infty.
		\end{equation}
		
		Since $\event{H}$ holds, and at the end of every episode we are using the OFU principle to select $\optpara{i}$, we have $\optcost{\para_\infty} \leq \optcost{\para_0}$. Hence, by Lemma \ref{OFU}, \eqref{Gprooflem4eq1} implies that $L_\infty$ is an optimal linear feedback for the true system $\para_0$. However, according to Lemma \ref{stabilizable}, $\Lmatrix{\para_0}$ is unique; i.e. $L_\infty=\Lmatrix{\para_0}$. Therefore, the limit is unique, which contradicts the divergence. Moreover, the convergence is to $\Lmatrix{\para_0}$; i.e.
		\begin{equation} \label{Gprooflem4eq3}
		\lim\limits_{i \to \infty} D_i = \para_0 \extendedLmatrix{\para_0} = D_0.
		\end{equation}
		Next, as shown in the proof of Lemma \ref{LipschitzK}, $\sum\limits_{\ell=0}^{\infty} {D_i}^\ell C {D_i'}^\ell$ is a Lipschitz function of $D_i$. Thus, plugging \eqref{Gprooflem4eq3} in \eqref{Gprooflem4eq2} we get
		\begin{equation*} 
		\lim\limits_{i \to \infty} \det \left(\frac{1}{n_i} V^{(i)}\right) = \det \left(\sum\limits_{\ell=0}^{\infty} {D_0}^\ell C {D_0'}^\ell\right),
		\end{equation*}
		which yields
		\begin{equation*} 
		\lim\limits_{i \to \infty} \det \left(\frac{1}{n_i} U^{(i)}\right) = \det \left( \extendedLmatrix{\para_0}\sum\limits_{\ell=0}^{\infty} {D_0}^\ell C {D_0'}^\ell \extendedLmatrix{\para_0}' \right).
		\end{equation*}
		Therefore, defining 
		\begin{equation*}
		\tilde{s}_j = \frac{\det \left( \sum\limits_{i=0}^{j} U_i \right)}{\det \left( \sum\limits_{i=0}^{j-1} U_i \right)},
		\end{equation*}
		we have
		\begin{equation*}
		\lim\limits_{j \to \infty} \tilde{s}_j = \lim\limits_{j \to \infty} \left(\frac{n_j}{n_{j-1}}\right)^q \frac{\det \left( \frac{1}{n_j}\sum\limits_{i=0}^{j} U_i \right)}{\det \left( \frac{1}{n_{j-1}} \sum\limits_{i=0}^{j-1} U_i \right)} = \lim\limits_{j \to \infty} \left(\frac{n_j}{n_{j-1}}\right)^q.
		\end{equation*}
		Note that according to \eqref{samplecoveigmax}, on $\event{G} \cap \event{H}$ the matrix $\frac{1}{n_i} U_i$ is bounded. Since the lengths of the episodes, $n_i, i=1,2,\cdots$, are growing exponentially, the matrix $\frac{1}{n_j} \sum\limits_{i=0}^{j} U_i$ is bounded as well. Using the definition of episode length in \eqref{algo2eq2}, we get
		\begin{equation} \label{Gprooflem4eq4}
		\lim\limits_{j \to \infty} \tilde{s}_j = \rrate.
		\end{equation}
		Finally, according to Lemma 11 in the work of Abbasi-Yadkori and Szepesv{\'a}ri \cite{abbasi2011regret},
		\begin{equation*}
		\sup\limits_{\para \neq 0} s_j \left(\para\right) \leq \tilde{s}_j.
		\end{equation*}
		So, \eqref{Gprooflem4eq4} implies the desired result.
	\end{proof}

\begin{proof}[\bf Proof of Lemma \ref{Gprooflem5}]
		Assuming $\event{G} \cap \event{H}$ holds, consider the following expression:
		\begin{equation*}
		\term{5} = \sum\limits_{t=1}^{T} \norm{\left( \para_t - \para_0 \right) \extendedLmatrix{\para_t} x(t)}{2}^2.
		\end{equation*}
		Since $\para_t$ does not change during each episode, we can write
		\begin{equation} \label{Gprooflem5eq1}
		\term{5} \leq \sum\limits_{j=1}^{\episodecount{T}} \sum\limits_{t=\lceil \tau_{j-1} \rceil }^{ \lceil \tau_j \rceil-1} \norm{\left( \optpara{j} - \para_0 \right) \extendedLmatrix{\optpara{j}} x(t)}{2}^2.
		\end{equation}
		Letting $\left\{ U_i \right\}_{i=0}^\infty$ be as defined in Lemma \ref{Gprooflem4}, $\sum\limits_{i=0}^{j} U_i$ is invertible and 
		\begin{eqnarray*}
			\sum\limits_{t=\lceil \tau_{j-1} \rceil }^{ \lceil \tau_j \rceil-1} \norm{ \left( \sum\limits_{i=0}^{j} U_i \right)^{-1/2} \extendedLmatrix{\optpara{i}} x(t)}{2}^2
			&=& \sum\limits_{t=\lceil \tau_{j-1} \rceil }^{ \lceil \tau_j \rceil-1} x(t)' \extendedLmatrix{\optpara{i}}' \left( \sum\limits_{i=0}^{j} U_i \right)^{-1} \extendedLmatrix{\optpara{i}} x(t) \\
			&=& \sum\limits_{t=\lceil \tau_{j-1} \rceil }^{ \lceil \tau_j \rceil-1} \tr{\left( \sum\limits_{i=0}^{j} U_i \right)^{-1} \extendedLmatrix{\optpara{i}} x(t)x(t)' \extendedLmatrix{\optpara{i}}'  }\\
			&=& \tr{\left( \sum\limits_{i=0}^{j} U_i \right)^{-1}U_j} \leq \tr{I_q}= q.	
		\end{eqnarray*}
		Further, using definition of $\left\{ s_j \left(\para \right) \right\}_{j=1}^\infty$ in Lemma \ref{Gprooflem4} we have
		\begin{eqnarray*}
			\Mnorm{\left( \optpara{j} - \para_0 \right) \left( \sum\limits_{i=0}^{j} U_i \right)^{1/2}}{2}^2
			&=& \Mnorm{ \left( \sum\limits_{i=0}^{j} U_i \right)^{1/2} \left( \optpara{j} - \para_0 \right)' \left( \optpara{j} - \para_0 \right) \left( \sum\limits_{i=0}^{j} U_i \right)^{1/2}}{2} \\
			&\leq& \tr{ \left( \optpara{j} - \para_0 \right) \sum\limits_{i=0}^{j} U_i \left( \optpara{j} - \para_0 \right)'} \\
			&\leq& p \Mnorm{\left( \optpara{j} - \para_0 \right) \sum\limits_{i=0}^{j} U_i \left( \optpara{j} - \para_0 \right)'}{2} \\
			&\leq& p \Mnorm{\left( \optpara{j} - \para_0 \right) \sum\limits_{i=0}^{j-1} U_i \left( \optpara{j} - \para_0 \right)'}{2} s_j \left( \optpara{j} - \para_0 \right).
		\end{eqnarray*}
		However, according to the definition of $\paraspace{j}$ in \eqref{algo2eq7}, both $\optpara{j}$, and $\para_0$ belong to $\bigcap\limits_{i=1}^{j-1} \tempparaspace{i}$. Therefore, \eqref{algo2eq6} implies
		\begin{eqnarray*}
			\Mnorm{\left( \optpara{j} - \para_0 \right) \sum\limits_{i=0}^{j-1} U_i \left( \optpara{j} - \para_0 \right)'}{2}
			&\leq& \sum\limits_{i=0}^{j-1} \Mnorm{\left( \optpara{j} - \para_0 \right) U_i \left( \optpara{j} - \para_0 \right)'}{2} \\
			&\leq& \Mnorm{\left(\optpara{j} - \para_0\right)'}{2}^2 + 4 \sum\limits_{i=1}^{j-1}  \prediction{n_i}{\frac{\delta}{i^2}} \\
			&\leq& 4 \ssconstant_{1}^2 + 4 \episodecount{T} \prediction{T}{\frac{\delta}{\episodecount{T}^2}},
		\end{eqnarray*}
		where in the last inequality above \eqref{boundedT}, $n_i \leq T$, and $i \leq \episodecount{T}$ are used. Putting everything together we have 
		\begin{eqnarray*}
			&&\sum\limits_{t=\lceil \tau_{j-1} \rceil }^{ \lceil \tau_j \rceil-1} \norm{\left( \optpara{j} - \para_0 \right) \extendedLmatrix{\optpara{j}} x(t)}{2}^2 \\
			&\leq& \Mnorm{\left( \optpara{j} - \para_0 \right) \left( \sum\limits_{i=0}^{j} U_i \right)^{1/2}}{2}^2 \sum\limits_{t=\lceil \tau_{j-1} \rceil }^{ \lceil \tau_j \rceil-1} \norm{ \left( \sum\limits_{i=0}^{j} U_i \right)^{-1/2} \extendedLmatrix{\optpara{i}} x(t)}{2}^2 \\
			&\leq& 4 pq s_j \left( \optpara{j} - \para_0 \right) \left( \ssconstant_{1}^2 + \episodecount{T} \prediction{T}{\frac{\delta}{\episodecount{T}^2}}\right).
		\end{eqnarray*}
		Plugging in \eqref{Gprooflem5eq1}, and using Lemma \ref{Gprooflem4}, leads to 
		\begin{equation} \label{Gprooflem5eq2}
		\term{5} \leq 4 pq \ssconstant_{6} \episodecount{T} \left( \ssconstant_{1}^2 + \episodecount{T} \prediction{T}{\frac{\delta}{\episodecount{T}^2}}\right).
		\end{equation}
		Going back to $\term{4}$, express it as
		\begin{equation*}
		\term{4} = \sum\limits_{t=1}^{T} x(t)' \extendedLmatrix{\para_t}' \left(\para_0+\para_t\right)' \Kmatrix{\para_t} \left(\para_0 - \para_t\right) \extendedLmatrix{\para_t} x(t). 
		\end{equation*}
		Letting $M_t=\Kmatrix{\para_t} \left(\para_0+\para_t\right) \extendedLmatrix{\para_t}$, the Cauchy-Schwartz inequality gives 
		\begin{eqnarray*}
			\term{4} &\leq& \sum\limits_{t=1}^{T} \norm{M_t x(t)}{2} \norm{\left(\para_0 - \para_t\right) \extendedLmatrix{\para_t} x(t)}{2}
			\leq \term{5}^{1/2} \left[ \sum\limits_{t=1}^{T} \norm{M_t x(t)}{2}^2 \right]^{1/2}. 
		\end{eqnarray*}
		By \eqref{ricatti3}, for all stabilizable $\para \in \R^{p \times q}$, the equality
		\begin{equation*}
		\extendedLmatrix{\para}' \para' \Kmatrix{\para} \para \extendedLmatrix{\para} = \Kmatrix{\para} - Q - \Lmatrix{\para}' R \Lmatrix{\para}
		\end{equation*}
		holds. Since $Q + \Lmatrix{\para_t}' R \Lmatrix{\para_t}$ is positive semidefinite,
		\begin{eqnarray}
		\eigmax{\extendedLmatrix{\para_t}' \para_t' \Kmatrix{\para_t} \para_t \extendedLmatrix{\para_t}} \leq \eigmax{\Kmatrix{\para_t}} \leq \ssconstant_{0} \label{Gprooflem5eq3}.
		\end{eqnarray}
		Moreover, for arbitrary $v \in \R^p$ we have
		\begin{equation*}
		\eigmin{R} \norm{\Lmatrix{\para_t}v}{2}^2 \leq v' \Lmatrix{\para_t}' R \Lmatrix{\para_t} v \leq v' \Kmatrix{\para_t} v \leq \ssconstant_{0} \norm{v}{2}^2,
		\end{equation*}
		which implies
		\begin{equation}
		\Mnorm{\extendedLmatrix{\para_t}}{2}^2 \leq 1+ \Mnorm{\Lmatrix{\para_t}}{2}^2 \leq 1+ \frac{\ssconstant_{0}}{\eigmin{R}} \label{Gprooflem5eq4}.
		\end{equation}
		Putting \eqref{Gprooflem4eq3}, \eqref{Gprooflem5eq4} together, we have
		\begin{equation*}
		\Mnorm{M_t}{2} \leq \Mnorm{\Kmatrix{\para_t} \para_t \extendedLmatrix{\para_t}}{2} + \Mnorm{\Kmatrix{\para_t} \para_0 \extendedLmatrix{\para_t}}{2} \leq \ssconstant_{0}+\ssconstant_{0} \Mnorm{\para_0}{2} \left(1+ \frac{\ssconstant_{0}}{\eigmin{R}}\right).
		\end{equation*}
		Thus, Lemma \ref{Wproofauxiliary} implies
		\begin{eqnarray*}
		\sum\limits_{t=1}^{T} \norm{M_t x(t)}{2}^2 \leq \ssconstant_5 \sum\limits_{i=1}^{\episodecount{T}} n_i \Mnorm{M^{(i)}}{2}^2 \leq T \episodecount{T} \ssconstant_{0}^2 \ssconstant_5 \left(1+ \Mnorm{\para_0}{2} \left(1+ \frac{\ssconstant_{0}}{\eigmin{R}}\right)\right)^2,
		\end{eqnarray*}
		which in addition to \eqref{Gprooflem5eq2} lead to the desired result.
	\end{proof}
	
	\begin{proof}[\bf Proof of Lemma \ref{Gprooflem6}]
		According to the definition of episode length in \eqref{algo2eq2}, we have
		\begin{eqnarray*}
			\tau_{i}-\tau_{i-1} &=& \rrate^{i/q} \left( \samplesize{\ref{stablemin}}{\frac{\eigmin{C}}{2}}{\frac{\delta}{i^2}}+1\right) \\
			&\geq& \rrate^{i/q} \left( \samplesize{\ref{stablemin}}{\frac{\eigmin{C}}{2}}{\delta}+1\right) \\
			&=& \rrate^{\frac{i-1}{q}} \tau_1.
		\end{eqnarray*}
		Since $\tau_{\episodecount{T}} \leq T$, we get
		\begin{equation*}
		T \geq \sum\limits_{i=1}^{\episodecount{T}} \left(\tau_i - \tau_{i-1}\right) \geq \frac{\rrate^{\frac{\episodecount{T}}{q}}-1}{\rrate^{\frac{1}{q}}-1} \tau_1,
		\end{equation*}
		which yields
		\begin{equation*}
		\episodecount{T} \leq \frac{q}{\log \rrate} \log \left( \frac{T \left(\rrate^{1/q}-1\right)}{\tau_1}+1\right).
		\end{equation*}
	\end{proof}

%\bibliographystyle{IEEEtran}
%\bibliography{References}

\end{document}